\documentclass[10pt]{article}

\usepackage{hyperref}

\usepackage[usenames,dvipsnames]{color}

\usepackage{amsmath}    													% need for subequations
\usepackage{graphicx}   													% need for figures
\usepackage{subfigure}														% need for subfigures
\usepackage{amssymb}    													% gives you \mathbb{} font
\usepackage[mathscr]{eucal}												% gives you \mathscr font
\usepackage{cancel}																% gives you the ability to visibly cross out terms in equations}
\usepackage{pstricks}
\usepackage{rotating}
\usepackage{lscape}
\usepackage[paperwidth=8.5in,paperheight=11in,top=1.25in, bottom=1.25in, left=1.00in, right=1.00in]{geometry}
\usepackage{mathtools}														% need for `show only references'
\mathtoolsset{showonlyrefs=true}									% only equations which are labeled AND referenced will be numbered.
\usepackage{fixltx2e,amsmath}                                           % Supposedly, this allows one to use \eqref{} in \caption{}.

\allowdisplaybreaks[1]

\usepackage{amsthm}											% need for theorem-proof environment
\newtheorem{theorem}{Theorem}
\newtheorem{lemma}[theorem]{Lemma}			% [theorem] ==> theorems and lemmas will share a counter
\newtheorem{proposition}[theorem]{Proposition}	
	
\theoremstyle{definition}
\newtheorem{remark}{Remark}

\numberwithin{equation}{section}	
\numberwithin{theorem}{section}

\newcommand{\<}{\left\langle} 
\renewcommand{\>}{\right\rangle}
\renewcommand{\(}{\left(}				
\renewcommand{\)}{\right)}
\renewcommand{\[}{\left[}
\renewcommand{\]}{\right]}			

\def\E{\mathbb{E}}			
\def\P{\mathbb{P}}										
\def\R{\mathbb{R}}

\def\ZZ{\mathbb{Z}}
\def\A{\mathscr{A}}
\def\C{\mathscr{C}}
\def\F{\mathscr{F}}

\def\L{\mathscr{L}}	
\def\M{\mathscr{M}}		
\def\O{\mathscr{O}}

\def\V{\mathscr{V}}
\def\D{\mathscr{D}}
\def\N{\mathscr{N}}
\def\S{\mathscr{S}}
\def\eps{\varepsilon}

\def\sig{\sigma}
\def\Lam{\Lambda}
\def\Gam{\Gamma}
\def\gam{\Delta}

\def\del{\delta}

\def\sigb{\bar{\sig}}
\def\gamb{\bar{\gam}}
\def\epsb{\bar{\eps}}

\def\Ph{\widehat{P}}

\def\Ih{\widehat{I}}

\def\d{\partial}		

\def\Ped{P^{\eps,\del}}
\def\Led{\L^{\eps,\del}}
\def\half{\frac{1}{2}}

\newcommand{\EEE}{\E^\star}
\def\It{\widetilde{I}}

\def\Pt{\widetilde{P}}
\def\Qt{\widetilde{Q}}

\newcommand{\PPP}{\P^{\star}}
\begin{document}

\title{Second Order Multiscale Stochastic Volatility Asymptotics: Stochastic Terminal Layer Analysis \& Calibration}

\author{
Jean-Pierre Fouque
\thanks{Department of Statistics \& Applied Probability,
 University of California,
        Santa Barbara, CA 93106-3110, \emph{fouque@pstat.ucsb.edu}. Work  supported by NSF grants DMS-0806461 and  DMS-1107468.}
\and Matthew Lorig
\thanks{Department of Applied Mathematics, University of Washington, Lewis Hall, Seattle, WA  98195, \emph{mlorig@uw.edu}.  Work partially supported by NSF grant DMS-0739195.}
\and Ronnie Sircar
\thanks{ORFE Department, Princeton University, Sherrerd Hall, Princeton NJ 08544, \emph{sircar@princeton.edu}. Work partially supported by NSF grant
    DMS-1211906.}}
\date{\today}
\maketitle

\begin{abstract}
Multiscale stochastic volatility models have been developed as an efficient way to capture the principle effects on derivative pricing and portfolio optimization of randomly varying volatility. The recent book Fouque, Papanicolaou, Sircar and S{\o}lna (2011, CUP) analyzes models in which the volatility of the underlying is driven by two diffusions -- one fast mean-reverting and one slow-varying, and provides a \emph{first order} approximation for European option prices and for the implied volatility surface, which is calibrated to market data. Here, we present the full \emph{second order} asymptotics, which are considerably more complicated due to a terminal layer near the option expiration time. We find that, to second order, the implied volatility approximation depends quadratically on log-moneyness, capturing the convexity of the implied volatility curve seen in data.  We introduce a new probabilistic approach to the terminal layer analysis needed for the derivation of the second order singular perturbation term, and calibrate to S\&P 500 options data.
\end{abstract}

%\tableofcontents

%%%%%%%%%%%%%%%%%%%%%%%%%%%%%%%%%%%%%%%
%
%				SECTION: Intro
%
%%%%%%%%%%%%%%%%%%%%%%%%%%%%%%%%%%%%%%%

\section{Introduction}
Stochastic volatility models relax the constant volatility assumption of the Black-Scholes model for option pricing by allowing volatility to fluctuate randomly.  In this context the market is incomplete in the sense that volatility is not traded and volatility risk cannot be fully hedged. There are many risk-neutral measures and we take the usual point of view that the market is choosing one of them by pricing call and put options for instance without introducing an arbitrage. As a result, stochastic volatility models are able to capture some of the well-known features of the implied volatility surface, such as the volatility smile and skew.  While some single-factor diffusion stochastic volatility models such as Heston's \cite{heston}, enjoy wide success due to the existence of semi-analytic pricing formula for European options, it is known that such models are not adequate to match implied volatility levels across all strikes and maturities; see, for instance, \cite{gatheral2}.
Numerous empirical studies have identified at least a fast time scale in stock price volatility on the order of days, as well as a slow scale on the order of months, for example \cite{chernov2003,fpss1,hillebrand,lebaron}.  This has motivated the development of multiscale stochastic volatility models, in which instantaneous volatility levels are controlled by multiple driving factors running on different time scales.

A class of multiscale stochastic volatility models is analyzed in \cite{fouque2004multiscale}, where an approximation for European options and their induced implied volatilities is derived, which can capture the overall level of implied volatility, its skew across strike prices and its term-structure over a wide range of maturities.  However, the analysis there is limited to a first order approximation, which cannot pick up the slight convexity of the observed equity implied volatility surface.
In this paper we extend the results of \cite{fouque2004multiscale} to second order.  This extension is non-trivial, as it requires a careful terminal layer analysis, which we approach probabilistically.  For some related multiscale perturbation techniques in European option pricing, we refer for instance to \cite{conlonsullivan} and \cite{lorig2} (spectral methods), \cite{sam1} (matched asymptotic expansions), \cite{alos1}, \cite{gobetmiri}  and \cite{fukasawa} (Malliavian calculus), \cite{fukasawa2} (Edgeworth expansion), and
\cite{zubellisouza} (inner-outer expansions). For a recent related analysis within a different asymptotic regime, see \cite{lorig-pagliarani-pascucci-2}.

Our second order results allow us to capture the slight convexity of the implied volatility skew.  Additionally, we are able to maintain analytic tractability which is important for calibration to data, as we demonstrate.
Of course, numerous asymptotic regimes have been analyzed in recent years for the option pricing problem in incomplete markets: see \cite{fpss}, \cite{gatheral} and \cite{lewis2000} for some references. Here our focus is not just on deriving and proving convergence of the approximation in the appropriate limits, but in disentangling the calibration procedure that results from it. Compared to the first order theory, this is much more involved as there are many more group parameters and basis functions that have to be accommodated to implied volatility data. Despite the increase in complexity, we show this can be implemented successfully.

The rest of this paper proceeds as follows.  In Section \ref{sec:model}, we describe the class of multiscale stochastic volatility models that we will work with.  Using a formal singular and regular perturbation analysis, we derive a pricing approximation which is valid for any European-style option.  We establish the accuracy of our pricing approximation in Theorem \ref{thm:accuracy}, where we use a regularization to handle the non-smoothness of payoffs such as call and put option payoffs.  In Section \ref{sec:volatilities}, we present an explicit formula for the implied volatility surface induced by our option pricing approximation.  Additionally, we show how a parameter reduction, crucial for calibration purpose, can be achieved with no loss of accuracy.  In Section \ref{sec:calibration}, we outline a procedure for calibrating the class of multiscale stochastic volatility models to the empirically observed implied volatility surface of liquid calls and puts.  We carry out this calibration procedure on S\&P500 index call and put options data.  Section \ref{sec:conclusion} concludes.

%%%%%%%%%%%%%%%%%%%%%%%%%%%%%%%%%%%%%%%
%
%				SECTION: Model + Pricing
%
%%%%%%%%%%%%%%%%%%%%%%%%%%%%%%%%%%%%%%%

\section{Second Order Option Pricing Asymptotics}
\label{sec:model}
We consider the class of multiscale stochastic volatility models studied in \cite{fpss}.  Let $X$ denote the price of a non-dividend-paying asset whose dynamics under the historical probability measure $\P$ is defined by the following system of stochastic differential equations (SDEs):
\begin{align}
\left. \begin{aligned}
dX_t
		&=		\mu \, X_t \, dt + f(Y_t,Z_t) \, X_t \, dW_t^{(0)} , \\
dY_t
		&=		\frac{1}{\eps} \alpha(Y_t) \, dt + \frac{1}{\sqrt{\eps}} \, \beta(Y_t) \, dW_t^{(1)} , \\
dZ_t
		&=		\del \, c(Z_t) \, dt + \sqrt{\del} \, g(Z_t) \, dW_t^{(2)}.
\end{aligned} \right\} \label{eq:Physical}
\end{align}
Here, ($W^{(0)}$, $W^{(1)}$, $W^{(2)}$) are $\P$-Brownian motions with correlation structure
$$
d\langle  W^{(0)},  W^{(1)} \rangle_t
		=		\rho_{1} \, dt ,\quad 
d\langle  W^{(0)},  W^{(2)} \rangle_t
		=		\rho_{2} \, dt ,	\quad
	d\langle  W^{(1)},  W^{(2)} \rangle_t
		=		\rho_{12} \, dt,	
		$$
where
 ($\rho_{1}$, $\rho_{2}$, $\rho_{12}$) satisfy $|\rho_{1}|, |\rho_{2}| ,|\rho_{12}| < 1$ and $1 + 2 \rho_{1} \rho_{2} \rho_{12} - \rho_{1}^2 - \rho_{2}^2 - \rho_{12}^2 > 0$, which guarantees that the correlation matrix of the Brownian motions is positive-semidefinite. The asset $X$ has geometric growth rate $\mu$ and stochastic volatility $f(Y_t,Z_t)$ which is driven by two factors, $Y$ and $Z$. Under the physical measure, the infinitesimal generators of $Y$ and $Z$ are scaled by factors of $1/\eps$ and $\del$ respectively.  Thus, $\eps>0$ and $1/\del>0$ represent the intrinsic time-scales of these processes.  We will work in the regime where $\eps << 1$ and $\del<<1$ so that $Y$ and $Z$ represent fast- and slow-varying factors of volatility respectively.  Most importantly, we assume the fast factor is \emph{mean-reverting}. Specifically,  $Y$ is an ergodic process, assumed reversible, and with a unique invariant distribution $\Pi$ under $\P$, which is independent of $\eps$. 

Under the risk-neutral pricing measure $\PPP$ (chosen by the market) the dynamics are described by
\begin{align}
\left. \begin{aligned}
dX_t
		&=		r \, X_t \, dt + f(Y_t,Z_t) \, X_t \, dW_t^{\star(0)} , 	
		\\
dY_t
		&=		\( \frac{1}{\eps} \alpha(Y_t) - \frac{1}{\sqrt{\eps}} \Lam(Y_t) \, \beta(Y_t) \) dt
					+ \frac{1}{\sqrt{\eps}} \, \beta(Y_t) \, dW_t^{\star(1)} , 	\\
dZ_t
		&=		\( \del \, c(Z_t) - \sqrt{\del} \, \Gam(Y_t,Z_t) \, g(Z_t) \) dt
					+ \sqrt{\del} \, g(Z_t) \, dW_t^{\star(2)} ,
\end{aligned} \right\} \label{eq:RiskNeutral}
\end{align}
where ($W^{\star(0)}$, $W^{\star(1)}$, $W^{\star(2)}$) are $\PPP$-Brownian motions with the same correlation structure as between their $\P$-counterparts, and $r {\geq 0}$ is the risk-free rate of interest.  The functions $\Lam(y)$ and $\Gam(y,z)$ represent market prices of volatility risk, which we have assumed such as to preserve the Markov structure of $(X,Y,Z)$,  the pair $(Y,Z)$, and $Y$ by itself.

\subsection{Assumptions}
\label{sec:assumptions}
Throughout this manuscript, we shall make the following assumptions which are stated here along with some of their immediate consequences essential to the paper: 
\begin{enumerate}
\item {For all starting points $(x,y,z)$, the systems of SDEs \eqref{eq:Physical} and \eqref{eq:RiskNeutral} have unique strong solutions} {$(X_t,Y_t,Z_t)$} for {all $0<\eps,\delta\leq1$.}  {Moreover,}  the coefficients are at most linearly growing.
\label{item:strong}

\item The volatility function $f$ of the two variables $(y,z)$ is measurable, bounded and bounded away from zero: there exist constants $\underline{c}$ and $\overline{c}$ such that $0<\underline{c}\leq f(y,z)\leq \overline{c}<\infty$ for all $(y,z)\in \R^2$.
\label{item:fbounded}

\item The market prices of volatility risk are bounded:  $|| \Lam ||_\infty <\infty$ and $|| \Gam ||_\infty < \infty $. \label{item:GammaLambdaBound} {In particular, combined with the previous assumption, $\P$ and $\PPP$ are equivalent and $\PPP$ is an Equivalent Martingale Measure.}

\item Let $Y^{(1)}$ be a diffusion process whose infinitesimal generator is {$\L_0:=\frac{1}{2} \beta^2(y) \d_{yy}^2 + \alpha(y) \d_y$} (so that, in distribution, $Y_t=Y^{(1)}_{t/\eps}$ under $\P$). We assume that $Y^{(1)}$ is a Feller process (that is, a Markov process with a Feller semigroup), that it is ergodic and its unique invariant distribution $\Pi$ has a strictly positive density denoted by $\pi$.
{
Furthermore, we assume the following specific {\it exponential ergodicity condition}: 
for every integer $k\geq 1$, there exist constants $c_k > 0$ and $d_k<\infty$ such that
 $$
 \L_0(y^{2k})\leq -c_ky^{2k}+d_k \qquad \forall y.
 $$
These conditions will enable us to use in Appendix \ref{AppA3} the exponential ergodic rates provided by Theorem 6.1 of \cite{MeynTweedie}.  
We note that two of the processes that are most commonly used as stochastic volatility drivers --- the Ornstein-Uhlenbeck (OU) and Cox-Ingersoll-Ross (CIR)  processes --- satisfy these conditions (in the case of CIR, the state space is $(0,\infty)$ with the classical condition on the coefficients ensuring that the process never hits zero).
}
\label{item:gap}
 
\item Let $Y^{(1,\eps)}$ be a diffusion process whose infinitesimal generator is {$\L_0-\sqrt{\eps}\Lambda(y)\beta(y) \d_y$} (so that, in distribution, $Y_t=Y^{(1,\eps)}_{t/\eps}$ under $\PPP$). We assume that $Y^{(1,\eps)}$ is a Feller process, that it is ergodic and its unique invariant distribution $\Pi_\eps$ has a strictly positive density denoted by $\pi_\eps$.
{ 
Furthermore, we assume the specific {\it exponential ergodicity condition}: 
for every integer $k\geq 1$, there exist constants $c_k> 0$ and $d_k<\infty$ independent of $\eps$ such that
 $$
\left[ { \L_0-\sqrt{\eps}\Lambda(y)\beta(y) \d_y} \right](y^{2k})\leq -c_ky^{2k}+d_k\qquad\forall y.
 $$ 
 Note that, for OU and CIR processes, this condition holds as a consequence of Assumptions \ref{item:GammaLambdaBound} and \ref{item:gap}.
}
\label{item:gapeps}
 
\item {The process $Y^{(1)}$ admits moments of any order uniformly bounded in $t<\infty$:}
\begin{align}
{ \sup_{t \geq 0}} \, \E \[ \left| Y_t^{(1)} \right|^k \] 	
	&\leq 		C(k).			\label{eq:Ybound}
\end{align}
Note that this assumption on moments is satisfied by OU and CIR processes (see \cite[Sections 3.3.3 and 3.3.4]{fpss} for more details on these processes).
 \label{item:Ybound}

\item Let $Z^{(1)}$ be a diffusion process whose infinitesimal generator is {$\M_2:=\frac{1}{2} g^2(z) \d_{zz}^2 + c(z) \d_z$} (so that, in distribution, $Z_t=Z^{(1)}_{\delta t}$ under $\P$). We assume that $Z^{(1)}$ admits moments of any order uniformly bounded in $t\leq T$, for fixed $T<\infty$:
\begin{align}
\sup_{t \leq T} \E \[ \left| Z_t^{(1)} \right|^k \] 
		&\leq 		C(T,k)	.\label{eq:Zbound}
\end{align}
\label{item:Zbound}
\item \label{item:Poisson} In addition to Assumption \ref{item:fbounded} ($f$ is bounded), we assume that {$f(y,\cdot) \in C^\infty(\R)$ for all $y\in \R$ with bounded derivatives.  
Note that consequently, the averaged square-volatility defined by 
\begin{align}
\sigb^2(z) 				&:= 		\int f^2(y,z) \, \Pi(dy) \label{eq:sigbar},
\end{align}
is finite and differentiable.}
Furthermore, consider Poisson equations of the form
\begin{align}
\L_0 \phi(\cdot,z) + \chi( \cdot, z)
	&= 0 , &
	&\text{where}&
\< \chi(\cdot,z) \>
	&:=	\int \chi(y,z) \Pi(dy) = 0 ,
\end{align}
and where $\chi$ is at most polynomially growing in $y$ and $z$.  
We assume solutions $\phi$ of such equations are at most polynomially growing in $y$ and $z$.  In particular, this applies to the solutions $\phi$ and $\{\psi_i, i=1,\dots,9\}$ to the Poisson equations \eqref{eq:phi}, \eqref{eq:psi1psi2} and \eqref{eq:psi3psi4}. In the cases that $Y$ is an OU or a CIR process, this follows from assumption \ref{item:fbounded} above and \cite[Lemmas 3.1 and 3.2]{fpss}.  
 \label{item:smooth}

\item \label{itm:h} We denote by $h: \R^+ \to \R$ the payoff function of a European option. 
The payoff $h$ is measurable, locally bounded ({\em i.e.} bounded on intervals $[a,b]$ for any $0<a,b<\infty$), and 
is at most polynomially growing at $0$ and $\infty$ (where, with a slight abuse of terminology, polynomially refers to inverse power law growth at $0$). In other words, there exist a finite constant $a\geq 0$ and an integer $k$ such that
\begin{align}
|h(x)|
	&\leq a(1+x^{k}+x^{-k}) , \qquad
\forall \, x 
	> 0 .
\end{align}
Note that $\log$-style payoffs (which are used to price variance swaps)
are in this class of payoffs, and of course it contains vanilla put and call payoffs (essential for calibration to implied volatilities), binary call and put payoffs, as well as other traded payoffs such as butterflies and straddles.

\begin{remark}\label{item:cases}
We will refer to $h$ as {\em smooth} in the case that $h\in C^\infty(0,\infty)$, and $h$ and all its derivatives are at most polynomially growing at $0$ and $\infty$.
The proof of accuracy for our second order pricing approximation (Theorem \ref{thm:accuracy}) will be separated into two parts.
First, in Appendix \ref{appendix:proof}, we establish the accuracy of the approximation for options with smooth payoffs. 
Results from the smooth case proof will be used in Appendix \ref{sec:call}, where we establish the accuracy of the approximation for options with payoffs which may have a finite number discontinuities in $h$ or its derivatives.
The proof that is given in Appendix \ref{sec:call} involves a regularization argument, which was used in  \cite{fouque2003proof} to establish the accuracy of the first order approximation with only a fast factor of volatility.
\end{remark}

\item
In what follows, we also assume that \eqref{eq:PDE+BC}, the linear pricing partial differential equation (PDE) given below, admits a unique classical solution.
\end{enumerate}

\subsection{Pricing PDE}
Consider a European option with expiration date $T$ and payoff $h(X_T)$.  The no-arbitrage pricing function of this option at time $t<T$ is given by the expectation of the discounted option payoff:
\begin{align}
\Ped(t,x,y,z)
		&=		\EEE \[ e^{-r(T-t)}h(X_T) \Big| X_t=x, Y_t=y, Z_t=z \] . 
\end{align}
Here, $\EEE$ denotes an expectation taken under the pricing measure $\PPP$, and we have used the Markov property of $(X,Y,Z)$.  The pricing function $\Ped$ {is the classical solution of} the following PDE and terminal condition:
\begin{align}
\Led \, \Ped
		&=		0 , &
\Ped(T,x,y,z)
		&=		h(x) , \label{eq:PDE+BC}
\end{align}
where, introducing the notation 
\begin{equation}
\D_k = x^k \d_{x \cdots x}^k,\qquad k=1,2,\cdots, \label{Ddef} 
\end{equation}
the operator $\Led$ is given by
\begin{align}
\Led
		&=		\( \frac{1}{\eps} \L_0 + \frac{1}{\sqrt{\eps}} \L_1 + \L_2 \)
					+ \sqrt{\del} \( \frac{1}{\sqrt{\eps}} \M_3 + \M_1 \) + \del \, \M_2 , \label{eq:L,eps,del}
\end{align}
with
\begin{align}
\L_0
		&=		\frac{1}{2} \beta^2(y) \d_{yy}^2 + \alpha(y) \d_y , \label{lzerodef}\\
\L_1
		&=		\rho_{1} \beta(y) f(y,z) \D_1 \d_y - \beta(y) \Lam(y) \d_y ,  \label{lonedef}\\
\L_2
		&=		\d_t  + \tfrac{1}{2} f^2(y,z) \D_2 + r \D_1 - r, \label{ltwodef}\\
\M_3
		&=		\rho_{12} \beta(y) g(z) \d_{yz}^2 ,  \label{mthreedef}\\
\M_1
		&=		\rho_{2} g(z) f(y,z) \D_1 \d_z - g(z) \Gam(y,z) \d_z , \\
\M_2
		&=		\frac{1}{2} g^2(z) \d_{zz}^2 + c(z) \d_z . \label{M2def}
\end{align}

For general coefficients $(f,\alpha,\beta,\Lam,c,g,\Gam)$, we do not have an explicit solution to \eqref{eq:PDE+BC}, and we seek an asymptotic approximation for the option price to make the calibration problem computationally tractable.  The fast factor asymptotic analysis is a singular perturbation problem, while the slow factor expansion is a regular perturbation. 
Thus, the small-$\eps$ and small-$\del$ regime gives rise to a combined singular-regular perturbation about the $\O(1)$ operator $\L_2$.  We expand $\Ped$ in powers of $\sqrt{\eps}$ and $\sqrt{\del}$ as follows
\begin{align}
\Ped(t,x,y,z)
		&=	\sum_{j \geq 0} \sum_{i \geq 0} \sqrt{\eps}^{\,i} \sqrt{\del}^{\,j} P_{i,j}(t,x,y,z) . \label{eq:Pexpand}
\end{align}
{This is a formal series expansion, for which we find $P_{i,j}$ for  $i+j\leq 2$ explicitly, and prove an accuracy result  for the truncated series in Section \ref{sec:proof}. As the combined regular-singular perturbation expansion is quite lengthy, we give a summary of the key results in Section \ref{summary}. 
We also point out that we are working within an infinite-dimensional family of models since the functions $(f,\alpha,\beta,\Lam,c,g,\Gam)$ are unspecified: the $18$ group parameters that are found in Section \ref{sec:parameters} and calibrated in Section \ref{sec:calibration} contain specific moments of these functions identified by the asymptotic analysis.
}

%		Regular Perturbation

\subsection{{Formal Asymptotics}}
{We first construct a regular perturbation expansion 
in powers of $\sqrt{\del}$ by writing}
\begin{align}
\Led
		&=		\L^\eps + \sqrt{\del} \, \M^\eps + \del \, \M_2 , &
\Ped
		&=		\sum_{j \geq 0} \sqrt{\del}^{\,j} P_j^\eps, \label{eq:L,P,eps}
\end{align}
where, from \eqref{eq:L,eps,del},
\begin{align}
\L^\eps
		&=		\frac{1}{\eps} \L_0 + \frac{1}{\sqrt{\eps}} \L_1 + \L_2 , &
\M^\eps
		&=		\frac{1}{\sqrt{\eps}} \M_3 + \M_1 , &
P_j^\eps
		&=		\sum_{i \geq 0} \sqrt{\eps}^{\,i} P_{i,j} . \label{eq:PepsExpand}
\end{align}
Inserting \eqref{eq:L,P,eps} into \eqref{eq:PDE+BC} and collecting terms of like-powers of $\sqrt{\del}$, we find that the lowest order equations of the regular perturbation expansion are
\begin{align}
\O(1):&&
0
		&=		\L^\eps P_0^\eps , \label{eq:del0} \\
\O(\sqrt{\del}):&&		
0
		&=		\L^\eps P_1^\eps + \M^\eps P_0^\eps , \label{eq:del1} \\
\O(\del):&&
0
		&=		\L^\eps P_2^\eps + \M^\eps P_1^\eps + \M_2 \, P_0^\eps . \label{eq:del2}
\end{align}
Within each of these three equations, we now perform a singular perturbation analysis with respect to $\eps$.

%					Singular Perturbation of O(1) Equation

\subsubsection{{First Order Fast Factor Term}\label{sec:one}}
From a fast factor expansion of equation \eqref{eq:del0}, we will now find {the} zeroth order term $P_{0,0}$ in our approximation \eqref{eq:Pexpand}, and the first term coming from the fast factor, $P_{1,0}$.  

We insert expansions \eqref{eq:PepsExpand} into \eqref{eq:del0} and collect terms of like-powers of $\sqrt{\eps}$.  
The resulting $\O(1/\eps)$ and $\O(1/\sqrt{\eps})$ equations are:
\begin{align}
\O(1/\eps):&&
0
		&=		\L_0 P_{0,0} , \\
\O(1/\sqrt{\eps}):&&
0
		&=		\L_0 P_{1,0} + \L_1 P_{0,0} .
\end{align}
We see from \eqref{lzerodef} and \eqref{lonedef} that all terms in $\L_0$ and $\L_1$ take derivatives with respect to $y$.  Thus, if we choose $P_{0,0}$ and $P_{1,0}$ to be independent of $y$, the above equations will automatically be satisfied.  Hence, we seek solutions of the form
\begin{align}
P_{0,0}
		&=		P_{0,0}(t,x,z) , &
P_{1,0}
		&=		P_{1,0}(t,x,z) ,
\end{align}
i.e., no $y$-dependence.  Continuing the asymptotic analysis, the $\O(1)$, $\O(\sqrt{\eps})$ and $\O(\eps)$ equations are:
\begin{align}
\O(1):&&
0
		&=		\L_0 P_{2,0} + \cancel{\L_1 P_{1,0}}+\L_2 P_{0,0} , \label{eq:u2poisson} \\
\O(\sqrt{\eps}):&&
0
		&=		\L_0 P_{3,0} + \L_1 P_{2,0} + \L_2 P_{1,0} , \label{eq:u3poisson} \\
\O(\eps):&&
0
		&=		\L_0 P_{4,0} + \L_1 P_{3,0} + \L_2 P_{2,0} , \label{eq:u4poisson}
\end{align}
where we have used the fact that $\L_1 P_{1,0}=0$.  

Equations \eqref{eq:u2poisson}, \eqref{eq:u3poisson} and \eqref{eq:u4poisson} are Poisson equations of the form
\begin{align}
0
		&=	\L_0 P +  \chi . \label{eq:poisson}
\end{align}
By the Fredholm alternative, equation \eqref{eq:poisson}, which is a linear ODE in $y$, admits a solution $P$ in $L^2(\Pi)$ only if the following solvability, or centering, condition holds:
\begin{align}
\< \chi \>  
		&:=			\int  \chi(y)\, \Pi(dy) 
		= 		0\, , \label{eq:center}
\end{align}
where we introduced the invariant distribution $\Pi$ in assumption \ref{item:gap} of Section \ref{sec:assumptions}. 
Note that two such solutions will differ by a constant (in $y$). 
We refer to \cite[Section 3.2]{fpss} for further details.

Applying the centering condition to equations \eqref{eq:u2poisson}, \eqref{eq:u3poisson} and \eqref{eq:u4poisson}, and using the fact that $P_{0,0}$ and $P_{1,0}$ do not depend on $y$, we find
\begin{align}
\O(1):&&
0
		&=		\< \L_2 \> P_{0,0} , \label{eq:center0} \\
\O(\sqrt{\eps}):&&
0
		&=		\< \L_1 P_{2,0} \> + \< \L_2 \> P_{1,0} , \label{eq:center1} \\
\O(\eps):&&
0
		&=		\< \L_1 P_{3,0} \> + \< \L_2 P_{2,0} \> , \label{eq:center2}
\end{align}
where, from \eqref{ltwodef}, the operator $\< \L_2 \>$ is given by
\begin{align}
\< \L_2 \>
		&=		\d_t + \tfrac{1}{2} \sigb^2(z) \D_2 + r \D_1 - r,    \label{eq:<L2>}
\end{align}  
with
\begin{equation}
		  \sigb^2(z):= \< f^2(\cdot,z)\>=\int f^2(y,z)\Pi(dy).
\end{equation} 
We observe that $\< \L_2 \> $ is the Black-Scholes pricing operator with {\it effective averaged} volatility 
$\sigb(z)$,  
in which the level $z$ of the slow factor appears as a parameter, {and we will express $P_{0,0}$ as a Black-Scholes option price in Proposition \ref{thm:prices}.}

Expanding the terminal condition in \eqref{eq:PDE+BC} leads to the terminal conditions
\begin{align}
\O(1):&&
P_{0,0}(T,x,z)
		&=		h(x) , \label{eq:u0BC} \\
\O(\sqrt{\eps}):&&
P_{1,0}(T,x,z)
		&=		0. \label{eq:u1BC} 
\end{align}
 
To find $P_{1,0}$ from equation \eqref{eq:center1}, we next compute $\< \L_1 P_{2,0} \>$. 
Using \eqref{eq:center0}, we re-write \eqref{eq:u2poisson} as follows
\begin{align}
\L_0 P_{2,0}
		=		- \L_2 P_{0,0} 
		=		- \( \L_2 - \< \L_2 \> \) P_{0,0} 
		=		-	\frac{1}{2}\( f^2 - \< f^2 \> \) \D_2 P_{0,0} .
\end{align}
Introducing a solution $\phi(y,z)$ to the Poisson equation
\begin{align}
\L_0 \, \phi
		&=		f^2 - \<f^2\>,  \label{eq:phi}
\end{align}
we deduce the following expression for $P_{2,0}$:
\begin{align}
P_{2,0}(t,x,y,z)
		&=		- \frac{1}{2} \, \phi(y,z)\, \D_2 P_{0,0}(t,x,z) + F_{2,0}(t,x,z) , \label{eq:P20,first}
\end{align}
{for some $F_{2,0}(t,x,z)$ that is independent of $y$, and which is yet to be determined.}  
Inserting \eqref{eq:P20,first} into \eqref{eq:center1} yields the following PDE for $P_{1,0}$ 
\begin{align}
\< \L_2 \> P_{1,0}
		&=		- \< \L_1 P_{2,0} \>
		=		- \< 	\bigg( \rho_{1} \beta \, f\, \D_1 \d_y - \beta \, \Lam \, \d_y \bigg)
								\( - \frac{1}{2} \phi \,\D_2 P_{0,0} + F_{2,0} \) \> 
		=		- \V \, P_{0,0} , \label{eq:u1PDE}
\end{align}
where the $z$-dependent operator $\V$ is given by
\begin{equation}
\V(z)
		=		V_3(z) \D_1 \, \D_2 + V_2(z) \D_2 , \label{Vdef}
\end{equation}
and we introduce the notation
\begin{align}\label{V2V3def}
V_2(z)
		&=		\frac{1}{2} \< \beta(\cdot) \Lam(\cdot) \d_y \phi(\cdot,z) \> , &
V_3(z)
		&=		- \frac12\rho_{1} \< \beta(\cdot) f(\cdot,z) \d_y \phi(\cdot,z) \> .
\end{align}
The solution $P_{1,0}$ of 
the PDE \eqref{eq:u1PDE} with terminal condition \eqref{eq:u1BC} will be given in Proposition \ref{thm:prices}.

\subsubsection{{Second Order Fast Factor Term $P_{2,0}$ and Terminal Layer}}\label{sec:P2}

The form of \eqref{eq:P20,first} shows that the natural terminal condition $P_{2,0}(T,x,y,z)=0$ is not enforceable because the singular perturbation with respect to the fast factor creates a terminal layer near $t=T$. However, as we will demonstrate in Section \ref{sec:proof}, the ergodic theorem enables us to impose the {\it averaged} terminal condition 
\begin{equation}
\< P_{2,0}(T,x,\cdot,z) \>
		=		0, \label{eq:u2BC}
\end{equation}
and to obtain the desired accuracy of our pricing approximation. In fact, we will see that this is the {\it only} appropriate choice for proof of convergence.
Moreover, the solution of the Poisson equation \eqref{eq:phi} is defined { in $L^2(\Pi)$} up to a constant in $y$. We choose this constant by imposing the condition 
\begin{align}
\<\phi(\cdot,z)\>&=0, \label{centeringphi}
\end{align} 
and we will show in Section \ref{sec:proof} that this choice 
{is needed in the proof of accuracy of our pricing approximation.}

{To determine $P_{2,0}$, given by \eqref{eq:P20,first}, we need a PDE and terminal condition for the unknown function $F_{2,0}$. 
These will be found from the centering conditions equation \eqref{eq:center2} and the terminal condition \eqref{eq:u2BC}.
Starting from the expression \eqref{eq:P20,first} for $P_{2,0}$, applying the operator $\L_2$ and averaging, we obtain: 
$$ \< \L_2 P_{2,0} \>
		=		\< \L_2 \( - \frac{1}{2} \phi \D_2 P_{0,0} + F_{2,0} \) \> 
		=		- \frac{1}{2} \< \phi \, \L_2 \> \D_2 P_{0,0} + \< \L_2 \> F_{2,0} . $$
Since $D_2$ and $\L_2$ commute when acting on functions independent of $y$, we have
$$ 	 \< \phi \, \L_2 \> \D_2 P_{0,0}= \D_2 \< \phi \, \L_2 \> P_{0,0}	= 	\D_2 \< \phi \, (\L_2-\<\L_2\>) \> P_{0,0} = \frac12D_2\< \phi f^2\>D_2P_{0,0}, $$
and therefore
\begin{equation}
 \< \L_2 P_{2,0} \>		=		A \, \D_2^2 \, P_{0,0} + \< \L_2 \> F_{2,0} , \label{eq:<L0u2>}
 \end{equation}
where $A(z)$ is given in \eqref{eq:Adef} below.}

To find $\< \L_1 P_{3,0} \> $ we first compute $P_{3,0}$.
From \eqref{eq:u3poisson}, \eqref{eq:center1}, \eqref{eq:phi}, \eqref{eq:P20,first}, and the definitions of $\L_1$ and $\L_2$, we have 
\begin{align}
\L_0 P_{3,0}
		&=		- \( \L_1 P_{2,0} + \L_2 P_{1,0} \) \\
		&=		- \( \L_1 P_{2,0} - \< \L_1 P_{2,0} \> \) - \( \L_2 - \< \L_2 \> \) P_{1,0} \\
		&=		- \L_1 \( - \frac{1}{2} \phi \D_2 P_{0,0} + F_{2,0} \) 
					+ \< \L_1 \( - \frac{1}{2} \phi \D_2 P_{0,0}+ F_{2,0} \) \> 
					- \( \frac{1}{2} \(f^2 - \< f^2 \> \) \D_2 P_{1,0} \) \\
		&=		- \(	-\frac{1}{2} \rho_{1} \Big( \beta f \d_y \phi - \< \beta f \d_y \phi \> \Big) \D_1 \D_2
								+\frac{1}{2} \Big( \beta \Lam \d_y \phi - \< \beta \Lam \d_y \phi\> \Big) \D_2 \) P_{0,0} 
					- \( \frac{1}{2} \L_0 \phi \) \D_2 P_{1,0} .
\end{align}
Therefore, we can write 
\begin{align}
P_{3,0}
		&=		\frac{1}{2} \rho_{1} \, \psi_1 \D_1 \D_2 P_{0,0} - \frac{1}{2} \psi_2 \D_2 P_{0,0}
					- \frac{1}{2} \phi \D_2 P_{1,0} + F_{3,0} , \label{eq:P3.explicit}
\end{align}
for some $F_{3,0}(t,x,z)$ which is independent of $y$, and where $\psi_1(y,z)$ and $\psi_2(y,z)$ satisfy the Poisson equations 
\begin{align}
\L_0 \, \psi_1
		&=		\beta f \d_y \phi - \< \beta f \d_y \phi \> , &
\L_0 \, \psi_2
		&=		\beta \Lam \d_y \phi - \< \beta \Lam \d_y \phi \>. \label{eq:psi1psi2}
\end{align}
Now, we can compute $\< \L_1 P_{3,0} \>$: 
\begin{align}
\< \L_1 P_{3,0} \> 
		&=		\< 	\Big(  \rho_{1} \beta f \D_1 - \beta \Lam \Big) \d_y
							\(  \frac{1}{2} \rho_{1} \, \psi_1 \D_1 \D_2 P_{0,0} 
									- \frac{1}{2} \psi_2 \D_2 P_{0,0}
									- \frac{1}{2} \phi \D_2 P_{1,0}
							\) \> {+\<\cancel{\L_1F_{3,0}}\>}\\
		&=		\Big( A_2 \D_1^2 \D_2 + A_1 \D_1 \D_2 + A_0 \D_2 \Big) \, P_{0,0} 
					+ \Big( V_3 \D_1 \D_2 + V_2 \D_2 \Big) \, P_{1,0} , \label{eq:<L-1u3>}
\end{align}
where $A_2(z)$, $A_1(z)$ and $A_0(z)$ are given in equation {\eqref{eq:Adef}} below.

Inserting \eqref{eq:<L0u2>} and \eqref{eq:<L-1u3>} into \eqref{eq:center2} yields the PDE for $F_{2,0}$ given in \eqref{eq:C2PDE} below. 
The terminal condition is found by averaging \eqref{eq:P20,first}, and using \eqref{eq:u2BC} and \eqref{centeringphi}:
\begin{align}
\< P_{2,0} (T,x,\cdot,z) \>
		&=		- \frac{1}{2} \cancel{ \< \phi \> } \D_2 P_{0,0}(T,x,z) + F_{2,0}(T,x,z) =  F_{2,0}(T,x,z) = 0,
\end{align}
{where we have used our choice on $\phi$ in equation \eqref{centeringphi}.}

In summary, we have that the function $F_{2,0}(t,x,z)$ satisfies the following PDE and terminal condition
\begin{align}
\< \L_2 \> F_{2,0}
		&=		- \A \, P_{0,0} - \V \, P_{1,0} , &
F_{2,0}(T,x,z)
		&=		0 , \label{eq:C2PDE} 
\end{align}
where the $z$-dependent operator $\A$ is given by
\begin{align}
\begin{aligned}
\A(z)
		&=		A_2(z) \D_1^2 \D_2 + A_1(z) \D_1 \D_2 + A_0(z) \D_2 + A(z) \D_2^2 ,  \\
A_2(z)
		&=	\frac{1}{2} \rho_{1}^2 \<\beta(\cdot) f(\cdot,z) \d_y \psi_1(\cdot,z) \> , \\
A_1(z)
		&=	- \frac{1}{2} \rho_{1} \( \< \beta(\cdot) \Lam(\cdot) \d_y \psi_1(\cdot,z) \>
				+ \< \beta(\cdot) f(\cdot,z) \d_y \psi_2(\cdot,z) \> \) , \\
A_0(z)
		&=	\frac{1}{2} \< \beta(\cdot) \Lam(\cdot) \d_y \psi_2(\cdot,z) \> , \\
{A(z)}
		&=		{- \frac{1}{4}  \< \phi(\cdot,z) f^2(\cdot,z) \>},		
\end{aligned} \label{eq:Adef}
\end{align}
The solution $F_{2,0}$ of the PDE with terminal condition \eqref{eq:C2PDE} will be given in Proposition \ref{thm:prices}.  This is as far as we will take the asymptotic analysis of the $\O(1)$ equation \eqref{eq:del0}.

%					Singular Perturbation of O(\sqrt{\del}) Equation

\subsubsection{{First Order Slow and Fast-Slow Terms $P_{0,1}$ and $P_{1,1}$}
\label{sec:one.half}}
Proceeding as in Section \ref{sec:one}, we insert expansions \eqref{eq:PepsExpand} into \eqref{eq:del1} and collect terms of like-powers of $\sqrt{\eps}$. The resulting $\O(\sqrt{\del}/\eps)$ and $\O(\sqrt{\del}/\sqrt{\eps})$ equations are:
\begin{align}
\O(\sqrt{\del}/\eps):&&
0
		&=		\L_0 P_{0,1} , \\
\O(\sqrt{\del}/\sqrt{\eps}):&&
0
		&=		\L_0 P_{1,1} + \L_1 P_{0,1} + \cancel{ \M_3 P_{0,0} } ,
\end{align}
where we have used $\M_3 P_{0,0}=0$ since $\M_3$, given in \eqref{mthreedef}, contains $\d_y$, and $P_{0,0}$ is independent of $y$.  Recalling that all terms in $\L_0$ and $\L_1$ also contain $\d_y$, we seek solutions $P_{0,1}$ and $P_{1,1}$ of the form
\begin{align}
P_{0,1}
		&=		P_{0,1}(t,x,z) , &
P_{1,1}
		&=		P_{1,1}(t,x,z) .
\end{align}
Continuing the asymptotic analysis, the $\O(\sqrt{\del})$ and  $\O(\sqrt{\del}\sqrt{\eps})$ equations are:
\begin{align}
\O(\sqrt{\del}):&&
0
		&=		\L_0 P_{2,1} + \cancel{\L_1 P_{1,1}} + \L_2 P_{0,1} 
					+ \cancel{\M_3 P_{1,0}} + \M_1 P_{0,0}, \label{eq:u21poisson} \\
\O(\sqrt{\del}\sqrt{\eps}):&&
0
		&=		\L_0 P_{3,1} + \L_1 P_{2,1} + \L_2 P_{1,1} 
					+ \M_3 P_{2,0} + \M_1 P_{1,0} . \label{eq:u31poisson}
\end{align}
Equations \eqref{eq:u21poisson} and \eqref{eq:u31poisson} are Poisson equations of the form \eqref{eq:poisson}.  Applying the centering condition \eqref{eq:center} to \eqref{eq:u21poisson} and \eqref{eq:u31poisson} yields
\begin{align}
\O(\sqrt{\del}):&&
0
		&=		\< \L_2 \> P_{0,1} + \< \M_1 \> P_{0,0}, \label{eq:u21center} \\
\O(\sqrt{\del}\sqrt{\eps}):&&
0
		&=		\< \L_1 P_{2,1} \> + \< \L_2 \> P_{1,1} 
					+ \< \M_3 P_{2,0} \> + \< \M_1 \> P_{1,0} . \label{eq:u31center}
\end{align}
We  also have the following terminal conditions
\begin{align}
\O(\sqrt{\del}):&&
P_{0,1}(T,x,z)
		&=		0 , \label{eq:u01BC} \\
\O(\sqrt{\del}\sqrt{\eps}):&&
P_{1,1}(T,x,z)
		&=		0 . \label{eq:u11BC}
\end{align}
The PDE \eqref{eq:u21center} and terminal condition \eqref{eq:u01BC} can be used to find an expression for $P_{0,1}$, which will be given in Proposition \ref{thm:prices}.  

The operator $\<\M_1\>$ appearing in \eqref{eq:u21center} can be written as 
\begin{align}
\< \M_1 \>
		= 		\rho_{2} g \< f \> \D_1 \d_z - g \< \Gam \> \d_z
		= \frac{2}{\sigb'} (V_1(z) \D_1\d_z + V_0(z)\d_z), \label{avM1def}
\end{align}
where ${\sigb'}=\partial_z\sigb$ (recall that we have assumed that $\sigb(z)$ in \eqref{eq:sigbar} is differentiable) and we introduce the notation
\begin{align}\label{V0V1def}
V_1(z) 
		&=		\frac{1}{2} \rho_{2} \sigb'(z) g(z) \< f(\cdot,z) \>, &
V_0(z)
		&=		- \frac{1}{2} \sigb'(z) g(z) \< \Gam(\cdot,z) \> .
\end{align}

In order to make use of equation \eqref{eq:u31center} to find $P_{1,1}$, we need expressions for $\< \L_1 P_{2,1} \>$ and $\< \M_3 P_{2,0} \>$.
To get to $\< \L_1 P_{2,1} \>$, we first compute $P_{2,1}$.
Using \eqref{eq:u21poisson} and \eqref{eq:u21center}, we have
\begin{align}
\L_0 P_{2,1}
		&=		- \L_2 P_{0,1} - \M_1 P_{0,0} \\
		&=		- \( \L_2 - \< \L_2 \> \) P_{0,1} - \( \M_1 - \< \M_1 \> \) P_{0,0} \\
		&=		- \frac{1}{2} \( f^2 - \<f^2\> \) \D_2 P_{0,1}
					- \rho_{2} g \( f - \<f\> \) \D_1 \d_z P_{0,0}
					+ g \( \Gam - \< \Gam \> \) \d_z P_{0,0}. 
\end{align}
Thus, $P_{2,1}$ is given by
\begin{align}
P_{2,1}
		&=		- \frac{1}{2} \phi \D_2 P_{0,1}
					- \rho_{2} g \psi_3 \D_1 \d_z P_{0,0}
					+ g \psi_4 \d_z P_{0,0} + F_{2,1}(t,x,z) , \label{eq:P21}
\end{align}
for some $F_{2,1}(t,x,z)$ which does not depend on $y$, and where $\psi_3(y,z)$ and $\psi_4(y,z)$ satisfy the Poisson equations 
\begin{align}
\L_0 \psi_3
		&=		f - \<f\> , &
\L_0 \psi_4
		&=		\Gam - \< \Gam \> . \label{eq:psi3psi4}
\end{align} 
Consequently,
\begin{align}
\< \L_1 P_{2,1} \>
		&=		\< \Big( 	\rho_{1} \beta f \D_1 - \beta \Lam \Big) \d_y
							\( 	- \frac{1}{2} \phi \D_2 P_{0,1} \) \> 
					+ \< \Big( 	\rho_{1} \beta f \D_1 - \beta \Lam \Big) \d_y
							\Big( - \rho_{2} g \, \psi_3 \D_1 \d_z P_{0,0} \Big) \>  \\ 
			&		+ \< \Big( 	\rho_{1} \beta f \D_1 - \beta \Lam \Big) \d_y	
							\Big( g \, \psi_4 \d_z P_{0,0} \Big) \> 
							+ \< \cancel{ \L_1 F_{2,1} } \> \\
		&=		- \frac{1}{2} \rho_{1} \< \beta f \d_y \phi \> \D_1 \D_2 P_{0,1}
					+ \frac{1}{2} \< \beta \Lam \d_y \phi \> \D_2 P_{0,1}  
					- \rho_{1} \rho_{2} g  \< \beta f \d_y \psi_3 \> \D_1^2 \d_z P_{0,0}\\
		& 					+ \rho_{2} g  \< \beta \Lam \d_y \psi_3 \> \D_1 \d_z P_{0,0}					+ \rho_{1} g  \< \beta f \d_y \psi_4 \> \D_1 \d_z P_{0,0}
					- g \< \beta \Lam \d_y \psi_4 \> \d_z P_{0,0} ,
\end{align}
which leads to  
\begin{align}
\< \L_1 P_{2,1} \>
		&=		\( V_3 \D_1 \D_2 + V_2 \D_2 \) P_{0,1}
					+ \frac{1}{\sigb'}\( C_2 \D_1^2 + C_1 \D_1 + C_0 \) \d_z P_{0,0} , \label{eq:L1P21} 
\end{align}
where $(C_0,C_1,C_2)$ are defined in \eqref{Cdefs} below.

Next, using expression \eqref{eq:P20,first} for $P_{2,0}$ we find
\begin{equation}
\< \M_3 P_{2,0} \>
=		\< \Big( \rho_{12} \beta(\cdot) g(z) \d_{yz}^2 \Big) \( - \frac{1}{2} \phi \D_2 P_{0,0} + F_{2,0} \) \> 
=		-\frac{1}{2} \rho_{12} g \< \beta \d_y \phi \> \D_2 \d_z P_{0,0} ,
\end{equation}
which gives
\begin{align}
\< \M_3 P_{2,0} \>
		&=		 \frac{1}{\sigb'} C \D_2 \d_z P_{0,0} , \label{eq:M3P20} 
\end{align}
where
\begin{align}\label{Cdefs}
C_2(z)
		&=		- \rho_{1} \rho_{2} \sigb'(z) g(z)  \< \beta(\cdot) f(\cdot,z) \d_y \psi_3(\cdot,z) \> , \\
C_1(z)
		&=		\rho_{2} \sigb'(z) g(z)  \< \beta(\cdot) {\Lam(\cdot)}\d_y \psi_3 \> 
					+ \rho_{1} g  \< \beta(\cdot) f(\cdot,z) \d_y \psi_4(\cdot,z) \> , \\
C_0(z)
		&=		- \sigb'(z) g(z) \< \beta(\cdot) \Lam(\cdot) \d_y \psi_4(\cdot,z) \> , \\
C(z)
		&=		-\frac{1}{2} \rho_{12} \sigb'(z) g(z) \< \beta(\cdot) \d_y \phi(\cdot,z) \> .
\end{align}

Inserting \eqref{eq:L1P21} and \eqref{eq:M3P20} into \eqref{eq:u31center}, we find
\begin{align}
\< \L_2 \> P_{1,1}
		&=	- \V \, P_{0,1} - \frac{1}{\sigb'} \C \, \d_z P_{0,0} - \< \M_1 \> P_{1,0} , \label{eq:P11PDE,first}
\end{align}
where the $z$-dependent operator $\C$ is given by
\begin{align}
\C(z)
		&=		C_2(z) \D_1^2 + C_1(z) \D_1 + C_0(z) + C(z) \D_2 . \label{Cdef}
\end{align}
The solution $P_{1,1}$ of the PDE \eqref{eq:P11PDE,first} with terminal condition \eqref{eq:u11BC} will be given in Proposition \ref{thm:prices} .  This is as far as we will take the asymptotic analysis of equation \eqref{eq:del1}

%					Singular Perturbation of O(\del) Equation

\subsubsection{{Second Order Slow Term}\label{Odel}}
We now move on to the $\O(\del)$ equation \eqref{eq:del2}.  Proceeding as in Sections \ref{sec:one} and \ref{sec:one.half}, we insert expansions \eqref{eq:PepsExpand} into \eqref{eq:del2} and collect term of like-powers of $\sqrt{\eps}$. The resulting $\O(\del/\eps)$ and $\O(\del/\sqrt{\eps})$ equations are:
\begin{align}
\O(\del/\eps):&&
0
		&=		\L_0 P_{0,2} , \\
\O(\del/\sqrt{\eps}):&&
0
		&=		\L_0 P_{1,2} + \L_1 P_{0,2} + \cancel{ \M_3 P_{0,1} } ,
\end{align}
where we have used $\M_3 P_{0,1}=0$ since $\M_3$ contains $\d_y$ and $P_{0,1}$ is independent of $y$.  Recalling that all terms in $\L_0$ and $\L_1$ also contain $\d_y$, we seek solutions $P_{0,2}$ and $P_{1,2}$ of the form
\begin{align}
P_{0,2}
		&= P_{0,2}(t,x,z) , &
P_{1,2}
		&= P_{1,2}(t,x,z) .
\end{align}
Continuing the asymptotic analysis, the $\O(\del)$ equation is:
\begin{align}
\O(\del):&&
0
		&=		\L_0 P_{2,2} + \cancel{\L_1 P_{1,2}} + \L_2 P_{0,2} 
					+ \cancel{\M_3 P_{1,1}} + \M_1 P_{0,1} + \M_2 P_{0,0} . \label{eq:u22poisson}
\end{align}
Equation \eqref{eq:u22poisson} is a Poisson equation of the form \eqref{eq:poisson} whose centering condition \eqref{eq:center} is
\begin{align}
\O(\del):&&
0
		&=		\< \L_2 \> P_{0,2} + \< \M_1 \> P_{0,1} + \M_2 P_{0,0}. \label{eq:u22center}
\end{align}
We also have the following terminal condition
\begin{align}
\O(\del):&&
P_{0,2}(T,x,z)
		&=		0 . \label{eq:u02BC}
\end{align}
The solution $P_{0,2}$ of the PDE \eqref{eq:u22center} with terminal condition \eqref{eq:u02BC} will be given in Proposition \ref{thm:prices}.  This is as far as we will take the combined singular-regular perturbation analysis.

%					Review of Important Equations

\subsection{Review of Asymptotic Analysis and Pricing Formulas\label{summary}}
In the previous sections we showed (formally) that the price of a European option can be approximated by
\begin{align}\label{eq:Phat} 
\Ped
		&\approx		\Pt^{\eps,\del}:=P_{0,0} + \sqrt{\eps} \, P_{1,0} + \sqrt{\del} \, P_{0,1} + \eps \, P_{2,0} + \del \, P_{0,2} 
												+ \sqrt{\eps\,\del} P_{1,1} ,
\end{align}
where
\begin{align}
\left. \begin{aligned}
\O(1):&&
\< \L_2 \> P_{0,0}
		&=		0 , &
P_{0,0}(T,x,z)
		&=		h(x) , \\
\O(\sqrt{\eps}):&&
\< \L_2 \> P_{1,0}
		&=		- \V P_{0,0} , &
P_{1,0}(T,x,z)
		&=		0 ,\\
\O(\sqrt{\del}):&&
\< \L_2 \> P_{0,1}
		&=		- \< \M_1 \> P_{0,0} , &
P_{0,1}(T,x,z)
		&=		0 , \\
\O(\eps):&&
P_{2,0}
		&=		- \frac{1}{2} \, \phi \D_2 P_{0,0} + F_{2,0} ,
\\ &&
\< \L_2 \> F_{2,0}
		&=		- \A \, P_{0,0} - \V \, P_{1,0} , &
F_{2,0}(T,x,z)
		&=		0 , \\
\O(\del):&&
\< \L_2 \> P_{0,2}
		&=		- \< \M_1 \> P_{0,1} - \M_2 P_{0,0} , &
P_{0,2}(T,x,z)
		&=		0 , \\
\O(\sqrt{\eps \, \del}):&&
\< \L_2 \> P_{1,1}
		&=		- \V \, P_{0,1} - \frac{1}{\sigb'} \C \, \d_z P_{0,0} - \< \M_1 \> P_{1,0} , &
P_{1,1}(T,x,z)
		&=		0 ,
\end{aligned} \right\} \label{eq:key}
\end{align}
and the $z$-dependent operators
in \eqref{eq:key} are given by
\begin{align}
\left. \begin{aligned}
\< \L_2 \>
		&= 		\d_t + \frac12\sigb^2 \D_2+ r \D_1 - r , \\
\V
		&= 		V_3 \D_1 \D_2 + V_2 \D_2, \\
\< \M_1 \>
		&=		\frac{2}{\sigb'} \( V_1 \D_1 + V_0 \) \d_z \\
\A
		&= 		A_2 \D_1^2 \D_2 + A_1 \D_1 \D_2 + A_0 \D_2 + A \D_2^2 , \\
\M_2
		&= 		\frac{1}{2} g^2 \, \d_{zz}^2 + c \, \d_z , \\
\C
		&= 		C_2 \D_1^2 + C_1 \D_1 + C_0 + C \D_2 .
\end{aligned} \right\} \label{eq:operators}
\end{align}

{
We introduce the Black-Scholes price of the option with volatility $\sigma$, time to maturity $\tau=T-t$, and payoff function $h$:
\begin{equation}
P_{BS}(\tau,x;\sigma) = e^{-r\tau}\int_\R h\(xe^{(r-\frac12\sigma^2)\tau + \sigma\sqrt{\tau}\,\xi}\)\frac{e^{-\xi^2/2}}{\sqrt{2\pi}}\,d\xi. \label{PBSdef}
\end{equation} 
Then we denote the solution to \eqref{eq:center0} with terminal condition \eqref{eq:u0BC} by
\begin{equation}
P_{0,0}(t,x,z) = P_{BS}(T-t,x;\sigb(z)), \label{P0formula}
\end{equation}
the Black-Scholes price with volatility $\sigb(z)$. }
In the following, we provide explicit expressions for the functions $P_{i,j}$ ($i+j\leq 2$) in terms of the contract's Black-Scholes price $P_{BS}$ and its derivatives (or ``Greeks''). 
\begin{proposition}
\label{thm:prices}
Let $\{P_{i,j},\, i+j\leq 2 \}$ be the unique classical solutions of the linear PDEs with terminal conditions given in \eqref{eq:key}.  Then we have the following expressions for the $\{ P_{i,j}\}$ {in terms of the Black-Scholes price $P_{BS}(T-t,x;\sigb(z))$ defined in \eqref{PBSdef}:
\begin{align}
&P_{0,0}(t,x,z)
		=		P_{BS}(T-t,x;\sigb(z)) , \qquad P_{1,0}(t,x,z)=\tau \, \V \, P_{BS},  \qquad P_{0,1}(t,x,z)=\tau \, \N_1 \d_\sig P_{BS}\label{P0}\\
&P_{2,0}(t,x,y,z)
		=		- \frac{1}{2} \, \phi(y,z)\D_2 P_{BS} + F_{2,0} , \quad \mbox{where}\quad 
F_{2,0}(t,x,z)
		=		\( \tau \, \A \, + \frac{1}{2} \tau^2 \V^2 \) P_{BS} , \label{P2}\\
&P_{0,2}(t,x,z)
		=		\( \frac{2\tau^{2}}{3\sigb'}  \N_1 \N_1' \, \d_\sig 
					+ \frac{\tau^{2}}{2}  \N_1^2 \( \d_{\sig\sig}^2 + \frac{1}{3\sigb} \d_\sig \)
					+	\frac{\tau}{3} B_2 \( \d_{\sig\sig}^2 + \frac{1}{2\sigb} \d_\sig \)
					+ \frac{\tau}{2} B_1 \, \d_\sig \) P_{BS} , \label{P02}\\
&P_{1,1}(t,x,z)
		=		\( \tau^2 \V \, \N_1 \, \d_\sig + \frac{\tau}{2} \C \, \d_\sig + \frac{\tau^2}{\sigb'} \N_1 \V' \) P_{BS} \label{P11}.
\end{align}}
Here $\tau= T-t$ is the time-to-maturity, and we have introduced the $z$-dependent operators
\begin{align}
\N_1
		&=		V_1\, \D_1 + V_0, &
\N_1'
		&=		V_1' \, \D_1 + V_0', &
\V'
		&= 		V_3' \D_1 \D_2 + V_2' \D_2 ,
\end{align}
and $z$-dependent parameters
\begin{align}\label{Bdefs}
V_j'&=\d_zV_j, \quad j=0,1,2,3 &
B_2
		&=		\frac{1}{2} g^2 (\sigb')^2 , &
B_1
		&=		\frac{1}{2} g^2 \sigb'' + c \sigb' ,
		\end{align}
where $(V_0(z), V_1(z), V_2(z), V_3(z))$ were defined in \eqref{V0V1def} and \eqref{V2V3def}.
\end{proposition}
{We re-iterate that all the terms are functions of $(t,x,z)$, except $P_{2,0}$, which also depends on the current level $y$ of the fast volatility factor. This is what creates the need for the terminal layer analysis in this paper.}

{In \eqref{P0formula}, we have already found that 
$P_{0,0} = P_{BS}(\sigb(z))$.}   
In order to derive expressions for the higher order terms $\{P_{i,j}, 1\leq i+j\leq 2\}$, we need the following two lemmas. 
\begin{lemma}[Vega-Gamma Relation]\label{thm:commute}
{The Black-Scholes pricing function $P_{BS}(\tau,x;\sigma)$ of a European option with time to maturity $\tau>0$ and payoff function $h$ satisfying Assumption \ref{itm:h} in Section \ref{sec:assumptions}, obeys the following relationship 
between its Vega $\d_\sigma P_{BS}$ and its Gamma $\D_2  P_{BS}$:}
\begin{align}
\d_\sigma P_{BS}(\tau,x;\sigma)
		&=		\tau  \sigma  \D_2  P_{BS}(\tau,x;\sigma) \quad \mbox{for all} \quad x>0 . \label{eq:Dsig}
\end{align}
\end{lemma}
\begin{proof}
{We have that 
$$ P_{BS}(\tau,x;\sigma)=e^{-r\tau}\int_{\R^+}h(y)\,p(\tau,x,y;\sig)\,dy, $$
where}
\begin{align}
p(\tau,x,y;\sig) = \frac{1}{y\sqrt{2\pi\sig^2\tau}}\exp\(-\frac{1}{2\sig^2\tau}\(\log(y/x)-(r-\frac12\sig^2)\tau\)^2\).
\end{align}
A direct computation shows that $\tau  \sigma  \D_2 p(\tau,x,y;\sig) = \d_\sig p(\tau,x,y;\sig)$.  Thus, we compute
\begin{align}
\tau  \sigma  \D_2  P_{BS}(\sigma)
	&=	e^{-r\tau}\tau \sig x^2 \d_{xx}^2 \int_{\R^+} p(\tau,x,y;\sig) h(y) \,dy 
	=	e^{-r\tau}\tau \sig x^2 \int_{\R^+} \d_{xx}^2 p(\tau,x,y;\sig) h(y) \,dy \\
	&=	e^{-r\tau}\int_{\R^+} \d_{\sig} p(\tau,x,y;\sig) h(y) \,dy
	=	\d_{\sig} \({e^{-r\tau}} \int_{\R^+} p(\tau,x,y;\sig) h(y) \,dy\) 
	=	\d_{\sig}  P_{BS}(\sigma) ,
\end{align}
where passing the derivative operators through the integrals is justified by the polynomial growth assumption (at $0$ and $+\infty$) on  the option payoff $h$. 
\end{proof}
\begin{remark}
Another way to derive the Vega-Gamma relationship \eqref{eq:Dsig} is to write a linear PDE with source for the Vega  $\d_\sigma P_{BS}(\sigma)$ by differentiating the Black-Scholes PDE for $P_{BS}(\sigma)$ and checking that the unique classical solution is given in terms of the Gamma by $\tau  \sigma  \D_2  P_{BS}(\sigma)$.
\end{remark}
Using Lemma \ref{thm:commute} and the fact that the logarithmic derivative operators $\D_k$ in \eqref{Ddef} commute
($\D_k \D_m = \D_m \D_k$), which implies that $\< \L_2 \>$ and any $\D_k$ commute ($\< \L_2 \>\D_k=\D_k\< \L_2 \>$), one can show:
\begin{lemma} \label{thm:recipe}
The Black-Scholes price $P_{BS}(\tau,x;\sigma)$ of a European option with time to maturity $\tau>0$, current stock price $x>0$, and payoff function $h$ satisfying 
Assumption \ref{itm:h} in Section \ref{sec:assumptions}, satisfies for positive integers $k$ and $n$,
\begin{align}
\< \L_2 \> \frac{\tau^{n+1}}{n+1} P(\{\D_k\}) \, P_{BS}(\tau,x;\sigma) 
		&=		- \tau^n \, P(\{\D_k\}) \, P_{BS}(\tau,x;\sigma)		, \label{eq:Dk} \\
\< \L_2 \> \frac{\tau^{n+1}}{n+2} P(\{\D_k\}) \, \d_\sig \, P_{BS}(\tau,x;\sigma)
		&=		- \tau^n \, P(\{\D_k\}) \, \d_\sig \, P_{BS}(\tau,x;\sigma)		 , \label{eq:DkDsig} \\
\< \L_2 \> \frac{\tau^{n+1}}{n+3} P(\{\D_k\}) \( \d_{\sig\sig}^2 + \frac{1}{\sig \, (n+2)} \d_\sig \) P_{BS}(\tau,x;\sigma)
		&=		- \tau^n \, P(\{\D_k\}) \, \d_{\sig\sig}^2 \, P_{BS}(\tau,x;\sigma)		, \label{eq:DkD2sig}
\end{align}
where $P(\{\D_k\})$ is some polynomial of $\D_1, \D_2, \cdots, \D_k$.
\end{lemma}
\begin{proof}
The proof is a straightforward calculation {of the left sides of the expressions \eqref{eq:Dk}, \eqref{eq:DkDsig} and\eqref{eq:DkD2sig}.} In showing the second and third relations, the $\d_\sigma$ partial derivatives acting on $P_{BS}$ are first converted into $\D_2$ using Lemma \ref{thm:commute} which now commute with any $\D_k$ operators and $\< \L_2 \>$. The final step uses that $\< \L_2 \>P_{BS}(\sigb(z))=0$.
\end{proof}

\begin{proof}[Proof of Proposition \ref{thm:prices}]
Using Lemmas \ref{thm:commute} and \ref{thm:recipe}, a direct computation shows that the $\{P_{i,j}\}$ of Proposition  \ref{thm:prices} satisfy the PDEs of \eqref{eq:key} and their associated terminal conditions.
\end{proof}

\subsection{Accuracy of the Approximation}\label{sec:proof}
The accuracy of our pricing approximation $\Pt^{\eps,\del}$ defined in \eqref{eq:Phat} is as follows.

\begin{theorem}
\label{thm:accuracy}
We recall the standing assumptions in Section \ref{sec:assumptions}.
Then, for fixed $t<T$, $x$, $y$, and $z$, the model price $P^{\eps,\del}$ solution of \eqref{eq:PDE+BC} and our price approximation, $\Pt^{\eps,\del}$ defined by \eqref{eq:Phat}, satisfy
\begin{align}
|P^{\eps,\del}(t,x,y,z)-\Pt^{\eps,\del}(t,x,y,z)|=\O(\eps^{3/2-}+\eps\sqrt{\del}+\del\sqrt{\eps}+\delta^{3/2}), 
\end{align}
where we use the notation $\O(\eps^{3/2-})$ to indicate terms that are of order $\O(\eps^{1+q/2})$ for any $q<1$.
\end{theorem}
\begin{proof}
The proof is divided into two parts.  First, in Appendix \ref{appendix:proof}, we provide a proof for options with smooth payoffs $h$, as described in Remark \ref{item:cases}.
Elements of the proof in the smooth case will be used to prove the accuracy for options with payoffs $h$ satisfying 
Assumption \ref{itm:h}, which is given in Appendix \ref{sec:call}.
\end{proof}

\begin{remark}[Terminal Layer Analysis]
\label{rmk:layer}
The main difficulty in Theorem \ref{thm:accuracy} in extending the accuracy of our pricing approximation from first order to second order is the treatment of the terminal condition for the second order term $P_{2,0}$ arising from the singular expansion due to the fast factor $Y$. In \cite{sam1}, the solution $P_{2,0}$ is derived by a formal matched asymptotic expansion with a terminal layer of size ${\eps}$. Here, in Appendix \ref{appendix:proof},  we provide a probabilistic proof for 
options with smooth payoffs $h$, which is based on the ergodic property of the fast factor $Y$, and justifies the choice of terminal condition made in \eqref{eq:u2BC}.  The proof of accuracy for options with payoffs $h$ satisfying Assumption \ref{itm:h}, is given in Appendix \ref{sec:call}.  
The proof makes use of the results derived in Appendix \ref{appendix:proof} and additionally relies on a payoff-regularization argument.
\end{remark}

\subsection{Group Parameters}
\label{sec:parameters}
We now summarize the parameters needed in the pricing approximation formulas derived in the previous section.
We begin by separating the $y$-dependent part in $\Pt^{\eps,\del}$ given by \eqref{eq:Phat}, by writing
\begin{align}
 \Pt^{\eps,\del}(t,x,y,z) = -\frac{1}{2} \, \eps \, \phi(y,z) \D_2 P_{0,0}(t,x,z) + \Qt^{\eps,\del}(t,x,z), 
 \end{align}
where
\begin{align}
\Qt^{\eps,\del}(t,x,z)
		&:=		P_{0,0} + \sqrt{\eps} \, P_{1,0} + \sqrt{\del} \, P_{0,1} + \sqrt{\eps \,\del} \, P_{1,1} + \eps \, F_{2,0}
									+ \del \, P_{0,2} .\label{expans}
\end{align}
  Using \eqref{eq:Phat}, \eqref{eq:key} and the linearity of the operator $\<\L_2\>$, we find that $\Qt^{\eps,\del}$ satisfies the following PDE and terminal condition
\begin{align}
\< \L_2 \> \Qt^{\eps,\del}
		&=		S^{\eps,\del} , &
\Qt^{\eps,\del}(T,x,z)
		&=		h(x) , \label{eq:L2Q=S}
\end{align}
where the source term $S^{\eps,\del}$ is given by
\begin{align}
S^{\eps,\del}
		=&		- \sqrt{\eps} \, \V P_{0,0} - \sqrt{\del} \, \< \M_1 \> P_{0,0} 
					- \sqrt{\eps \,\del} \, \( \V \, P_{0,1} + \frac{1}{\sigb'} \C \, \d_z P_{0,0} + \< \M_1 \> P_{1,0}\) \\ &
					- \eps \, \Big( \A \, P_{0,0} + \V \, P_{1,0} \Big) - \del \, \Big( \< \M_1 \> P_{0,1} + \M_2 P_{0,0}\Big)\\
		=&		- (\sqrt{\eps} \, \V) P_{0,0} - (\sqrt{\del} \, \< \M_1 \>) P_{0,0} 
					- (\sqrt{\eps}\,\V)(\sqrt{\del}\,P_{0,1}) - (\sqrt{\eps\del} \,\C) \, \frac{1}{\sigb'}\d_z P_{0,0} - (\sqrt{\del}\< \M_1 \> )(\sqrt{\eps}\,P_{1,0}) \\ &
					- (\eps \A)  P_{0,0} - (\sqrt{\eps}\,\V) (\sqrt{\eps}\,P_{1,0})  - (\sqrt{\del}\< \M_1 \>) (\sqrt{\del}\,P_{0,1}) - (\del \M_2) P_{0,0}.
\end{align}
To extract which group parameters are needed for the price expansion, we
absorb a half-integer power of $\eps$ and/or $\del$ into the corresponding group parameters and define:
\begin{align}
V_i^\eps			&:=		\sqrt{\eps} \, V_i, &
V_i^\del			&:=		\sqrt{\del} \, V_i, &
A_i^\eps								&:= 	\eps \, A_i, &
B_i^\del								&:= 	\del \, B_i, &
C_i^{\eps,\del}	&:=		\sqrt{\eps\del} \, C_i, \label{sizing}
\end{align}
{where the $V_i$ were defined in \eqref{V0V1def} and \eqref{V2V3def}, and the $A_i$, $B_i$ and $C_i$ in \eqref{eq:Adef}, \eqref{Bdefs} and \eqref{Cdefs} respectively. }
Similarly, we absorb the appropriate $\eps$ or $\delta$ pre-multiplier into the terms of the expansion
\eqref{expans} by defining $P_{1,0}^\eps$ and $P_{0,1}^\delta$ through
$$  \sqrt{\eps} \, P_{1,0}(t,x,z)=P_{1,0}^\eps(t,x;\sigb(z),V_2^\eps(z),V_3^\eps(z)) , \qquad  \sqrt{\delta}P_{0,1}(t,x,z)=P_{0,1}^\delta(t,x;\sigb(z),V_0^\del(z),V_1^\del(z)). $$
Substituting from \eqref{eq:operators} the expressions for $\M_2,\V,\A,\<\M_1\>$ and $\C$, and changing the $\d_z$ derivatives in $\<\M_1\>$ and $\M_2$ acting on $P_{0,0}$ into $\d_\sig$ derivatives acting on $P_{BS}(\sigb(z))$, we finally have
\begin{align}					
S^{\eps,\del}=	&- \(V_3^\eps \D_1 \D_2 + V_2^\eps \D_2 \) P_{BS} - 2\(V_1^\del \D_1 + V_0^\del\) \d_\sig P_{BS} \\ 
			&- \(V_3^\eps \D_1 \D_2 + V_2^\eps \D_2 \)  \, P_{0,1}^\del
			- \( C_2^{\eps,\del}\D_1^2 + C_1^{\eps,\del}\D_1 + C_0^{\eps,\del}+ C^{\eps,\del}\D_2 \)	\d_\sig P_{BS} \\ 
			&- 2\(V_1^\del \D_1 + V_0^\del\)\(\d_\sig + \frac{{V_3'}^{\eps}}{\sigb'} \d_{V_3^\eps} 
							+ \frac{{V_2'}^{ \eps}}{\sigb'} \d_{V_2^\eps} \) P_{1,0}^\eps \\ 
			&- \Big( A_2^\eps \D_1^2 \D_2 + A_1^\eps \D_1 \D_2 + A_0^\eps \D_2 + A^\eps \D_2^2 \Big) \, P_{BS}
			- \(V_3^\eps \D_1 \D_2 + V_2^\eps \D_2 \) P_{1,0}^\eps \\ 
			&- 2\(V_1^\del \D_1 + V_0^\del\) \(\d_\sig + \frac{{V_1'}^{\del}}{\sigb'} \d_{V_1^\del}
							+ \frac{{V_0'}^{\del}}{\sigb'} \d_{V_0^\del} \) P_{0,1}^\del - 
							\( B_2^\del \d_{\sig\sig}^2 + B_1^\del \d_\sig \) P_{BS} .
\end{align}
Here our notation is ${V_i'}^\eps(z)=\d_zV_i^\eps(z)$, and similarly ${V_i'}^{\del}$.
Since $P_{1,0}$ is linear in $V_3$ and $V_2$ and $P_{0,1}$ is linear in $V_1$ and $V_0$, neither $\d_{V_3^\eps} P_{1,0}$, $\d_{V_2^\eps} P_{1,0}$, $\d_{V_1^\del} P_{0,1}$ nor $\d_{V_0^\del} P_{0,1}$ contain any of the $V_i$'s (that is, they are order one quantities).  

As such, the {\it group parameters} that appear in the source term $S^{\eps,\del}$ and therefore, in the price approximation \eqref{eq:Phat} are
\begin{align}
V_3^\eps, V_2^\eps, V_1^\del, V_0^\del, C_2^{\eps,\del}, C_1^{\eps,\del}, C_0^{\eps,\del}, C^{\eps,\del}, A_2^\eps, A_1^\eps, A_0^\eps, A^\eps, B_2^\del, B_1^\del, \frac{{V_3'}^\eps}{\sigb'}, \frac{{V_2'}^\eps}{\sigb'}, \frac{{V_1'}^\del}{\sigb'}, \frac{{V_0'}^\del}{\sigb'} .
\label{eq:unobservables}
\end{align}
These 18 parameters, which move with the slow volatility factor $Z_t$, as well as $\phi^\eps(y,z) := \eps \, \phi(y,z)$ needed in \eqref{eq:Phat}, can be obtained by calibrating the class of multiscale stochastic volatility models to the implied volatility surface of (liquid) European options, as described the Section \ref{sec:calibration}. Note from \eqref{sizing} that the $V_i^\eps$  are order $\sqrt{\eps}$, the $V_i^\del$ order $\sqrt{\del}$ and that they appeared in the first order asymptotic theory in \cite{fouque2004multiscale}. The new parameters $(A_i^\eps,B_i^\del,C_i^{\eps,\del})$ come from the order $\eps$, order $\del$ and order $\sqrt{\eps\del}$ terms in the the second order expansion respectively.  

\subsubsection{Parameter Reduction}
\label{sec:parameter reduction}
The group parameters in \eqref{eq:unobservables}
 depend on the current
level $z$ of the slow volatility factor and, in the case of $\phi^\eps$, on the fast factor too. 
In order to calibrate completely from the implied volatility surface and not use historical returns data to estimate $\sigb(z)$, we replace it by a quantity $\sig^*(z)$ which absorbs the term $V_2^\eps(z)$. In so doing, there is now one less parameter (listed explicitly for calibration purposes in \eqref{parameters*}), and we show {in Appendix \ref{sec:reduction}} that the accuracy of the second order approximation is unchanged.

{We define
\begin{align}
\sig^*(z) := \sqrt{\sigb(z)^2 + 2  V_2^\eps(z)} ,  \label{eq:replacement}
\end{align}
and $P_{i,j}^*$ as the solutions to 
\begin{align}
\left. \begin{aligned}
\O(1):&&
\< \L_2^* \> P_{0,0}^*
		&=		0 , &
P_{0,0}^*(T,x,z)
		&=		h(x) , \\
\O(\sqrt{\eps}):&&
\< \L_2^* \> P_{1,0}^*
		&=		- \V^* P_{0,0}^* , &
P_{1,0}^*(T,x,z)
		&=		0 ,\\
\O(\sqrt{\del}):&&
\< \L_2^* \> P_{0,1}^*
		&=		- \< \M_1 \> P_{0,0}^* , &
P_{0,1}^*(T,x,z)
		&=		0 , \\
\O(\eps):&&
P_{2,0}^*
		&=		- \frac{1}{2} \, \phi \D_2 P_{0,0}^* + F_{2,0}^* ,
\\ &&
\< \L_2^* \> F_{2,0}^*
		&=		- \A \, P_{0,0}^* - \V^* \, P_{1,0}^* , &
F_{2,0}^*(T,x,z)
		&=		0 , \\
\O(\del):&&
\< \L_2^* \> P_{0,2}^*
		&=		- \< \M_1 \> P_{0,1}^* - \M_2 P_{0,0}^* , &
P_{0,2}^*(T,x,z)
		&=		0 , \\
\O(\sqrt{\eps \, \del}):&&
\< \L_2 \> P_{1,1}^*
		&=		- \V^* \, P_{0,1}^* - \frac{1}{\sigb'} \C \, \d_z P_{0,0}^* - \< \M_1 \> P_{1,0}^* , &
P_{1,1}^*(T,x,z)
		&=		0 ,
\end{aligned} \right\} \label{eq:key2}
\end{align}
where
\begin{align}
\< \L_2^* \>
		&:=		\< \L_2 \> + \sqrt{\eps} V_2 \D_2 , &
\V^*
		&:=		\V - V_2 \D_2 .
\end{align}
These correspond to the PDEs and terminal conditions in \eqref{eq:key} of the asymptotic approximation to second order with $\sigb(z)$ replaced by $\sig^*(z)$, and the terms containing $V_2$ removed. Their solutions are exactly as in Proposition \ref{thm:prices} with $\sig^*(z)$ in place of $\sigb(z)$ and both $V_2$ and $V_2'$ set to zero.}

\begin{proposition}[Parameter Reduction]
\label{thm:reduction}
For payoff functions $h$ as described in Assumption \ref{itm:h}, the price approximation 
\begin{align}
P^{*,\eps,\del}:=P_{0,0}^* + \sqrt{\eps} \, P_{1,0}^*  + \sqrt{\del} \, P_{0,1}^*  + \eps \, P_{2,0}^*  + \del \, P_{0,2}^* + \sqrt{\eps\,\del} P_{1,1} ^* , \label{Pstardef}
\end{align}
has the same accuracy as obtained in Theorems \ref{thm:accuracy}:
\begin{align}
|P^{\eps,\del}(t,x,y,z)-P^{*,\eps,\del}(t,x,y,z)|&=  \O(\eps^{3/2-}+\eps \sqrt{\del}+ \del\sqrt{\eps} + \delta^{3/2}) . 
\end{align}
\end{proposition}
\begin{proof}
The proof is given in Appendix \ref{sec:reduction}.
\end{proof}

%%%%%%%%%%%%%%%%%%%%%  Pricing Exotic Options %%%%%%%%%%%%%%%%%%%%%%%%%%%%%%%%

%%%%%%%%%%%%%%%%%%%%%%%%%%%%%%%%%%%%%%%%%%%%%%%
%
%  Proof of Accuracy / Boundary Layer Analysis
%
%%%%%%%%%%%%%%%%%%%%%%%%%%%%%%%%%%%%%%%%%%%%%%

%%%%%%%%%%%%%%%%%%%%%%%%%%%%%%%%%%%%%%%
%
%				SECTION: Implied Volatilities
%
%%%%%%%%%%%%%%%%%%%%%%%%%%%%%%%%%%%%%%%

\section{Asymptotics for Implied Volatilities and Calibration}\label{sec:volatilities} 
It is common practice to quote option prices in units of implied volatility, by inverting the Black-Scholes formula for European call options with respect to the volatility parameter. This does not imply that the Black-Scholes assumptions of constant volatility are adopted, it is merely a convenient change of unit through which to view the departure of market data from the Black-Scholes theory, and to assess improvements due to multiscale stochastic volatility as we use here. In what follows, we translate the second order expansion of options prices found in the previous section, to a corresponding expansion in implied volatility units.

\subsection{Implied Volatility Expansion}
\label{sec:imp.vol}
{We seek an implied volatility expansion of the form}
\begin{align}
I^{\eps,\del}
		&=		\sum_{j \geq 0} \sum_{i \geq 0} \sqrt{\eps}^{\,i} \sqrt{\del}^{\,j} I_{i,j} 
		&\text{such that} &&
\Ped
		&=		P_{BS}\( I^{\eps,\del} \) .
\end{align}
Performing a Taylor expansion of $P_{BS}(I^{\eps,\del})$ about  $I_{0,0}$ and rearranging terms yields
\begin{align}
	P_{0,0} + \sqrt{\eps} \, P_{1,0} + \sqrt{\del} \, P_{0,1} &+ \sqrt{\eps \,\del} \, P_{1,1} + \eps \, P_{2,0}
									+ \del \, P_{0,2} + \cdots \\ 
		&=		P_{BS}(I_{0,0} + \sqrt{\eps} \, I_{1,0} + \sqrt{\del} \, I_{0,1} + \sqrt{\eps \,\del} \, I_{1,1} + \eps \, I_{2,0} + \del \, I_{0,2} + \cdots ) \\ 
		&=		P_{BS}(I_{0,0}) + \sqrt{\eps} \, I_{1,0} \d_\sig P_{BS}(I_{0,0}) + \sqrt{\del} \, I_{0,1} \d_\sig P_{BS}(I_{0,0}) \\ 
		&		\qquad + \sqrt{\eps \, \del} \, \Big( I_{1,0} I_{0,1} \d_{\sig\sig}^2 P_{BS}(I_{0,0}) 
											+ I_{1,1} \d_\sig P_{BS}(I_{0,0}) \Big) \\
				&\qquad + \eps \( \frac{1}{2} I_{1,0}^2 \d_{\sig\sig}^2 P_{BS}(I_{0,0}) + I_{2,0} \d_\sig P_{BS}(I_{0,0}) \) \\
				&\qquad+ \del \( \frac{1}{2} I_{0,1}^2 \d_{\sig\sig}^2 P_{BS}(I_{0,0}) + I_{0,2} \d_\sig P_{BS}(I_{0,0}) \) + \cdots.
				\label{eq:Ped=Pbs}
\end{align}
Equating terms in \eqref{eq:Ped=Pbs} of like powers in the parameters $\eps$ and $\del$, and using $P_{0,0}=P_{BS}(\sigb)$, we find
\begin{align}
\left. \begin{aligned}
\O(1):&&
I_{0,0}
		&=		\sigb , &
\O(\eps):&&
I_{2,0}
		&=		\frac{P_{2,0}}{\d_\sig P_{0,0}} - \frac{1}{2} I_{1,0}^2 \frac{\d_{\sig\sig}^2 P_{0,0}}{\d_\sig P_{0,0}} , \\		
\O(\sqrt{\eps}):&&
I_{1,0}
		&=		\frac{P_{1,0}}{\d_\sig P_{0,0}} , &
\O(\del):&&
I_{0,2}
		&=		\frac{P_{0,2}}{\d_\sig P_{0,0}} - \frac{1}{2} I_{0,1}^2 \frac{\d_{\sig\sig}^2 P_{0,0}}{\d_\sig P_{0,0}} , \\		
\O(\sqrt{\del}):&&
I_{0,1}
		&=		\frac{P_{0,1}}{\d_\sig P_{0,0}} , &
\O(\sqrt{\eps\del}):&&
I_{1,1}
		&=		\frac{P_{1,1}}{\d_\sig P_{0,0}} - I_{1,0} I_{0,1} \frac{\d_{\sig\sig}^2 P_{0,0}}{\d_\sig P_{0,0}} .
\end{aligned} \right\} \label{eq:Iij}
\end{align}
For a European call or put option with strike price $K$ and time to maturity $\tau$ it is convenient to express  the $I_{i,j}$'s as  functions of {forward log-moneyness}
\begin{align}
d
		&:=		\log \( K / x e^{r \tau} \) &
\text{(forward log-moneyness).}
\end{align}
Setting the payoff function $h(x) = (x - K)^{+}$ for a call option  and using the expressions given for $\{ P_{i,j} \}$ in Theorem \ref{thm:prices},  the $I_{i,j}$'s in \eqref{eq:Iij} become
\begin{flalign}
\O(1):&&
I_{0,0}
		&=		\sigb , \label{eq:Icall}\\
\O(\sqrt{\eps}):&&
I_{1,0}
		&=		V_2 \frac{1}{\sigb}
					+ V_3 \( \frac{1}{2 \, \sigb} + \frac{d}{\tau \, \sigb^3} \) , \\
\O(\sqrt{\del}):&&
I_{0,1}
		&=		V_0 \, \tau
					+ V_1 \( \frac{\tau}{2} + \frac{d}{\sigb^2} \) , \\
\O(\eps):&&
I_{2,0}
		&=		\frac{-\phi}{2 \, \tau \, \,\sigb}
					+ V_2^2 \( -\frac{1}{2 \,\sigb^3} \)
					+ V_2 V_3 \( -\frac{3 d}{\tau \,\sigb^5}-\frac{1}{2 \,\sigb^3} \) \\ && &\qquad
					+	V_3^2 \( -\frac{3 d^2}{\tau^2 \,\sigb^7}+\frac{3}{2 \tau \,\sigb^5}-\frac{3 d}{2 \tau \,\sigb^5} \) \\ && &\qquad
					+ A \( \frac{d^2}{\tau^2 \,\sigb^5}-\frac{1}{\tau \,\sigb^3}-\frac{1}{4 \,\sigb} \)
					+ A_0 \( \frac{1}{\sigb} \) 
					+ A_1 \( \frac{d}{\tau \,\sigb^3}+\frac{1}{2 \,\sigb} \) \\ && &\qquad
					+ A_2 \( \frac{d^2}{\tau^2 \,\sigb^5}-\frac{1}{\tau \,\sigb^3}+\frac{d}{\tau \,\sigb^3}+\frac{1}{4 \,\sigb} \) , \\
\O(\del):&&
I_{0,2}
		&=		V_0^2 \( \frac{\tau^2}{6 \,\sigb} \)
					+ V_0 V_1 \( - \frac{5 \,d \,\tau}{3 \,\sigb^3}+\frac{\tau^2}{6 \,\sigb} \)
					+ V_1^2  \( -\frac{7 \, d^2}{3 \,\sigb^5}+\frac{5 \tau}{6 \,\sigb^3}- \frac{5 \, d \,\tau}{6 \,\sigb^3}
							+\frac{\tau^2}{6 \,\sigb} \) \\&& &\qquad
					+ V_0 \frac{V_0'}{\sigb'} \( \frac{2 \tau^2}{3} \)
					+ V_0 \frac{V_1'}{\sigb'} \( \frac{\tau^2}{3}+\frac{2 \, d\, \tau}{3 \,\sigb^2} \)
					+ V_1 \frac{V_0'}{\sigb'} \( \frac{\tau^2}{3}+\frac{2 \, d \, \tau}{3 \,\sigb^2} \) \\&& &\qquad
					+ V_1 \frac{V_1'}{\sigb'} \( \frac{\tau^2}{6}+\frac{2 \, d^2}{3 \,\sigb^4}-\frac{2 \tau}{3 \,\sigb^2}
							+\frac{2 d\,\tau}{3 \,\sigb^2} \)
					+ B_2 \( \frac{d^2}{3 \,\sigb^3}+\frac{\tau}{6 \,\sigb}-\frac{\tau^2 \,\sigb}{12} \)
					+ B_1 \( \frac{\tau}{2} \) , \\
\O(\sqrt{\eps\,\del}):&&
I_{1,1}
		&=		V_0 V_2 \( -\frac{\tau}{\,\sigb^2} \)
					+ V_0 V_3 \( - \frac{3 \, d}{\,\sigb^4}-\frac{\tau}{2 \,\sigb^2} \) \\ && &\qquad
					+ V_1 V_2 \( - \frac{3 \, d}{\,\sigb^4}-\frac{\tau}{2 \,\sigb^2} \)
					+ V_1 V_3 \( -\frac{6 \, d^2}{\tau \,\sigb^6}+\frac{3}{\,\sigb^4}-\frac{3 \, d}{\sigb^4} \) \\ && &\qquad
					+ V_0 \frac{V_2'}{\sigb'} \( \frac{\tau}{\,\sigb} \)
					+ V_0 \frac{V_3'}{\sigb'} \( \frac{d}{\,\sigb^3}+\frac{\tau}{2 \,\sigb} \)
					+ V_1 \frac{V_2'}{\sigb'} \( \frac{d}{\,\sigb^3}+\frac{\tau}{2 \,\sigb} \) \\ && &\qquad
					+ V_1 \frac{V_3'}{\sigb'} \( \frac{d^2}{\tau \,\sigb^5}-\frac{1}{\,\sigb^3}+\frac{d}{\sigb^3}+\frac{\tau}{4 \,\sigb} \)
					+ C_2 \( \frac{\tau}{8}+\frac{d^2}{2 \tau \,\sigb^4}-\frac{1}{2 \,\sigb^2}+\frac{d}{2 \,\sigb^2} \) \\ && &\qquad
					+ C_1 \( \frac{\tau}{4}+\frac{d}{2 \,\sigb^2} \)
					+ C_0 \( \frac{\tau}{2} \)
					+ C \(  -\frac{\tau}{8}+\frac{d^2}{2 \tau \,\sigb^4}-\frac{1}{2 \,\sigb^2}\) .
\end{flalign}

Observe that this second order expansion produces an implied volatility curve which is {\it quadratic} in log-moneyness $d$ and therefore accounts for the slight turn in the skew that is most prominent in shorter maturity options data, as we will see in Figure \ref{fig:ImpVolFit2006}. The first order approximation derived in \cite{fouque2004multiscale} is linear in $d$ and therefore only accounted for the skew effect.
Note also that the parameter reduction outlined in {Section \ref{sec:parameter reduction}} can be applied to this implied volatility expansion as well ($\sigb$ replaced by $\sig^*$ and $V_2$-terms removed), and this will be used in the calibration in the next section.
We also remark that  the formal second order expansion for  the case of a
\emph{single slow} volatility factor had previously been
considered in \cite{FLT}, \cite{roger} and \cite{smile},  for instance.

%%%%%%%%%%%%%%%%%%%%%%%%%%%%%%%%%%%%%%%%%%%%%%%%%%%%%%%%%
%
%		Calibration
%
%%%%%%%%%%%%%%%%%%%%%%%%%%%%%%%%%%%%%%%%%%%%%%%%%%%%%%%%%

\subsection{Calibration} \label{sec:calibration}
In this section we discuss how the  parameters \eqref{eq:unobservables}, 
can be obtained by calibrating the multiscale class of models to liquid European options data.
We define
\begin{align}
\It^{\eps,\del}
		&:=		I_{0,0} + \sqrt{\eps} \, I_{1,0} + \sqrt{\del} \, I_{0,1} + \sqrt{\eps \,\del} \, I_{1,1} + \eps \, I_{2,0}
									+ \del \, I_{0,2} . \label{eq:Ihat2}
\end{align}
Using \eqref{eq:Icall} and the parameter reduction described in Proposition \ref{thm:reduction}, we have
\begin{align}
\It^{\eps,\del}
		&=		\( \frac{1}{\tau} \, k + l + \tau \, m + \tau^2 \, n \) 
					+ \frac{d}{\tau} \( p + \tau \, q + \tau^2 \, s \) 
					+ \frac{d^2}{\tau^2} \( u + \tau \, v + \tau^2 \, w \) , \label{eq:Imodel2}
\end{align}
where 
\begin{align}
\O(1/\tau):&& k &=
		\frac{3 (V_3^\eps)^2}{2 (\sigma^*)^5}-\frac{A_2^\eps}{(\sigma^*)^3}-\frac{A^\eps}{(\sigma^*)^3}-\frac{\phi^\eps}{2 {\sigma^*}} , \label{eq:k2w2}\\
\O(1):&& l &=
		\frac{3 V_1^\del V_3^\eps}{(\sigma^*)^4}
		-\frac{C_2^{\eps,\del}}{2 (\sigma^*)^2}-\frac{C^{\eps,\del}}{2 (\sigma^*)^2} \\ && &\qquad
		+\frac{A_0^\eps}{{\sigma^*}}+\frac{A_1^\eps}{2 {\sigma^*}}+\frac{A_2^\eps}{4 {\sigma^*}}-\frac{A^\eps}{4 {\sigma^*}}-\frac{V_1^\del {V_3'}^\eps}{(\sigma^*)^3 {\sigma^*}'}
		+{\sigma^*} + \frac{V_3^\eps}{2{\sigma^*}} , \\
\O(\tau):&& m &=
		\frac{B_1^\del}{2}+\frac{C_0^{\eps,\del}}{2}+\frac{C_1^{\eps,\del}}{4}+\frac{C_2^{\eps,\del}}{8}-\frac{C^{\eps,\del}}{8}+\frac{5 (V_1^\del)^2}{6 (\sigma^*)^3} \\ && &\qquad
		-\frac{V_0^\del V_3^\eps}{2 (\sigma^*)^2}
		+\frac{B_2^\del}{6 {\sigma^*}}-\frac{2 V_1^\del {V_1'}^\del}{3 (\sigma^*)^2 {\sigma^*}'} 
		+\frac{V_0^\del {V_3'}^\eps}{2 {\sigma^*} {\sigma^*}'}
		+\frac{V_1^\del {V_3'}^\eps}{4 {\sigma^*} {\sigma^*}'} + V_0^\del + \frac{V_1^\del}{2} , \\
\O(\tau^2):&& n &=
		\frac{(V_0^\del)^2}{6 {\sigma^*}}+\frac{V_0^\del V_1^\del}{6 {\sigma^*}}+\frac{(V_1^\del)^2}{6 {\sigma^*}}-\frac{B_2^\del {\sigma^*}}{12}+\frac{2 V_0^\del {V_0'}^\del}{3 {\sigma^*}'} \\ && &\qquad
		+\frac{{V_0'}^\del V_1^\del}{3 {\sigma^*}'}+\frac{V_0^\del {V_1'}^\del}{3 {\sigma^*}'}+\frac{V_1^\del {V_1'}^\del}{6 {\sigma^*}'} , \\
\O(d/\tau):&& p &=
		- \frac{3 (V_3^\eps)^2}{2 (\sigma^*)^5} + \frac{A_1^\eps}{(\sigma^*)^3} + \frac{A_2^\eps}{(\sigma^*)^3} 
		+ \frac{V_3^\eps}{(\sigma^*)^3}, \\
\O(d):&& q &=
		-\frac{3 V_0^\del V_3^\eps}{(\sigma^*)^4}-\frac{3 V_1^\del V_3^\eps}{(\sigma^*)^4}
		+\frac{C_1^{\eps,\del}}{2 (\sigma^*)^2}+\frac{C_2^{\eps,\del}}{2 (\sigma^*)^2} \\ && &\qquad
		+\frac{V_0^\del {V_3'}^\eps}{(\sigma^*)^3 {\sigma^*}'}
		+\frac{V_1^\del {V_3'}^\eps}{(\sigma^*)^3 {\sigma^*}'}
		+ \frac{V_1^\del}{(\sigma^*)^2}, \\
\O(d \, \tau):&& s &=
		-\frac{5 V_0^\del V_1^\del}{3 (\sigma^*)^3}-\frac{5 (V_1^\del)^2}{6 (\sigma^*)^3}+\frac{2 {V_0'}^\del V_1^\del}{3 (\sigma^*)^2 {\sigma^*}'}+\frac{2 V_0^\del {V_1'}^\del}{3 (\sigma^*)^2 {\sigma^*}'}
		+\frac{2 V_1^\del {V_1'}^\del}{3 (\sigma^*)^2 {\sigma^*}'} , \\
\O(d^2/\tau^2):&& u &=
		-\frac{3 (V_3^\eps)^2}{(\sigma^*)^7}+\frac{A_2^\eps}{(\sigma^*)^5}+\frac{A^\eps}{(\sigma^*)^5}, \\
\O(d^2/\tau):&& v &=
		-\frac{6 V_1^\del V_3^\eps}{(\sigma^*)^6}+\frac{C_2^{\eps,\del}}{2 (\sigma^*)^4}+\frac{C^{\eps,\del}}{2 (\sigma^*)^4}+\frac{V_1^\del {V_3'}^\eps}{(\sigma^*)^5 {\sigma^*}'} , \\
\O(d^2):&& w &=
		-\frac{7 (V_1^\del)^2}{3 (\sigma^*)^5}+\frac{B_2^\del}{3 (\sigma^*)^3}+\frac{2 V_1^\del {V_1'}^\del}{3 (\sigma^*)^4 {\sigma^*}'} .
%\end{aligned} \right\} 
\end{align}
In total, we have ten ``basis functions'' with which to fit the empirically observed implied volatility surface:
\begin{align}
\left\{ \frac{1}{\tau},1,\tau,\tau^2,\frac{d}{\tau},d,d\tau,\frac{d^2}{\tau^2},\frac{d^2}{\tau},d^2 \right\}.
\end{align}
It will be  helpful to define
\begin{align}
\Theta
		&:= \{k,l,m,n,p,q,s,u,v,w \} , \\
\Phi
		&:= \{\sig^*,V_3^\eps, V_1^\del, V_0^\del, C_2^{\eps,\del}, C_1^{\eps,\del}, C_0^{\eps,\del}, C^{\eps,\del}, A_2^\eps, A_1^\eps, A_0^\eps, A^\eps, B_2^\del, B_1^\del, \frac{{V_3'}^\eps}{\sigb'}, \frac{{V_1'}^\del}{\sigb'}, \frac{{V_0'}^\del}{\sigb'},\phi^\eps \} .\label{parameters*}
\end{align}
We let $I(\tau,d)$ be the implied volatility of a European call option with time to maturity $\tau$ and forward log-moneyness $d$ as observed from option prices on the market.  We let $\Ih^{\eps,\del}(\tau,d;\Theta)$ be the implied volatility of a European call as calculated using \eqref{eq:Imodel2}.  The calibration procedure consists of the following steps:
\begin{enumerate}
\item Find $\Theta^{*}$ such that
\begin{align}
\min_{\Theta} \sum_i \sum_j \( I(\tau_i,d_j) - \Ih^{\eps,\del}(\tau_i,d_j;\Theta) \)^2
		&=		\sum_i \sum_j \( I(\tau_i,d_j) - \Ih^{\eps,\del}(\tau_i,d_j;\Theta^{*}) \)^2 ,
\end{align}
where the double sum runs over all maturities $\tau_i$ and strikes $K_j$ (corresponding to forward log-moneyness $d_j$) for which a call or put is liquidly traded. This is the least-squares fit of formula \eqref{eq:Imodel2} resulting in estimated $k,l,m,\cdots,w$.
\item Next the ten constraints of equation \eqref{eq:k2w2} are used to find the minimal $L_2$ set of parameters $\Phi^*$. That is, we find $\Phi^{*}$ such that
\begin{align}
\min_{\Phi \in \mathscr{I}} \left\| \Phi \right\|^2 
		&= 		\left\| \Phi^{*} \right\|^2 , &
\mathscr{I}
		&=		\left\{ \Phi : \text{equation \eqref{eq:k2w2} holds with $\Theta = \Theta^{*}$}\right\} .
\end{align}
\end{enumerate}
We emphasize that our calibration procedure encompasses all maturities, that is we do not fit maturity-by-maturity.
Note that the implied volatility approximation $\widetilde{I}^{\eps,\del}$, defined in \eqref{eq:Ihat2}, retains the same order of accuracy as the price approximation $\Pt^{\eps,\del}$ in the case of a non-smooth payoff.
This follows directly from smoothness of the Black-Scholes formula as a function of the volatility.

%		Data

\subsection{Data\label{datasec}}
We perform the described calibration procedure on European call and put options on the S\&P500 index on two separate dates, one pre-crisis on October 19, 2006, and one post-crisis on March 18, 2010.    
In Figure \ref{fig:ImpVolFit2006} we plot the implied volatility fit from October 19, 2006.  The parameters obtained from the above calibration procedure are
\begin{align}
\sig^* &= 0.2051, &
V_3^\eps &= -0.0034, & 
V_1^\del &= 0.0023,  &
V_0^\del &= -0.0064,  &
C_2^{\eps,\del} &= -0.0073, & 
C_1^{\eps,\del} &= -0.0171, \\
C_0^{\eps,\del} &= 0.0183,  &
C^{\eps,\del} &= 0.0047,  &
A_2^\eps &= -0.0002,  &
A_1^\eps &= 0.0038,  &
A_0^\eps &= -0.0183,  &
A^\eps &= 0.0011, \\
B_2^\del &= 0.0080,  &
B_1^\del &= 0.0183,  &
\frac{{V_3'}^\eps}{\sigb'} &= 0.0146, & 
\frac{{V_1'}^\del}{\sigb'} &= -0.3104,  &
\frac{{V_0'}^\del}{\sigb'} &= 0.9856, &
\phi^\eps  &= -0.0181 .
\end{align}
In Figure \ref{fig:ImpVolFit2010} we plot the implied volatility fit from March 18, 2010.  The parameters obtained from the above calibration procedure are
\begin{align}
\sig^* &= 0.2269, &
V_3^\eps &= -0.0062, & 
V_1^\del &= -0.0026, & 
V_0^\del &= 0.0208, & 
C_2^{\eps, \del} &= -0.0031, & 
C_1^{\eps, \del} &= -.00034, \\
C_0^{\eps, \del} &= -0.0035, & 
C^{\eps, \del} &= 0.0033, & 
A_2^\eps &= 0.0034, & 
A_1^\eps &= 0.0034, & 
A_0^\eps &= -0.0004, & 
A^\eps &= -0.0012, \\
B_2^\del &= 0.0012, & 
B_1^\del &= -0.0035, & 
\frac{{V_3'}^\eps}{\sigb'} &= -0.1590, & 
\frac{{V_1'}^\del}{\sigb'} &= 0.0914, &
\frac{{V_0'}^\del}{\sigb'} &= -0.0729, &
\phi^\eps  &= -0.0443.
\end{align}
Notice that, in both cases, the obtained parameters other than $\sig^*$ are small, as expected in the regime of validity of our expansion (i.e., small $\eps$ and small $\del$). 

%%%%%%%%%%%%%%%%%%%%%%%%%%%%%%%%%%%%%%%
%
%				SECTION: Figures
%
%%%%%%%%%%%%%%%%%%%%%%%%%%%%%%%%%%%%%%%

%\clearpage
\begin{figure}
\centering
\begin{tabular}{ | c | c |}
\hline
\includegraphics[width=0.5\textwidth,height=0.23\textheight]{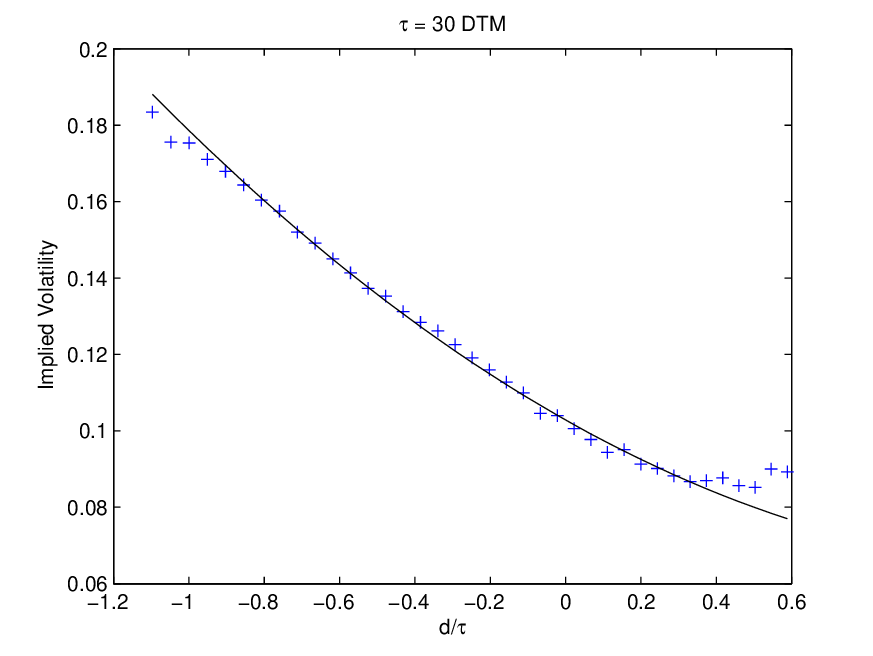} &
\includegraphics[width=0.5\textwidth,height=0.23\textheight]{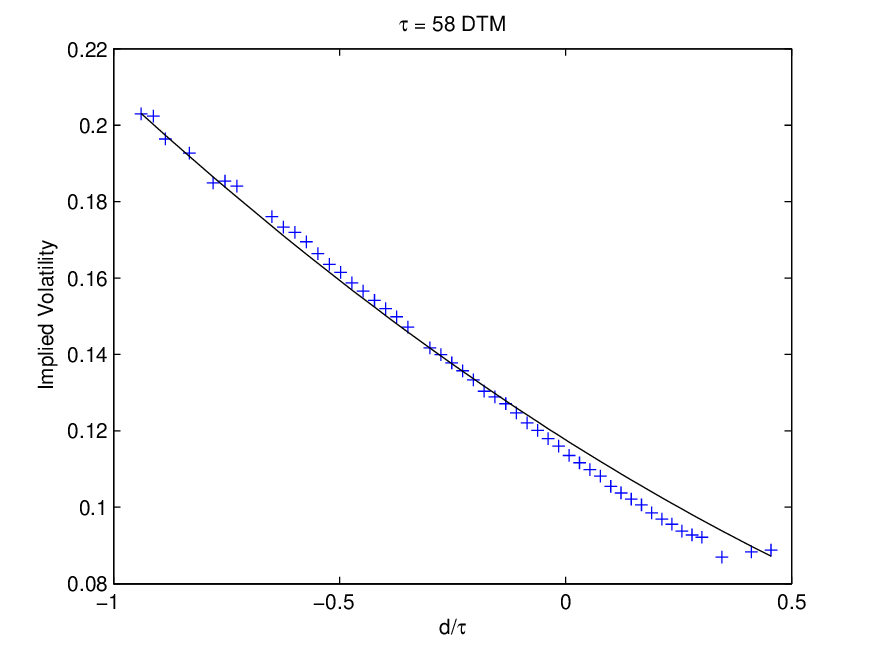} \\ \hline
\includegraphics[width=0.5\textwidth,height=0.23\textheight]{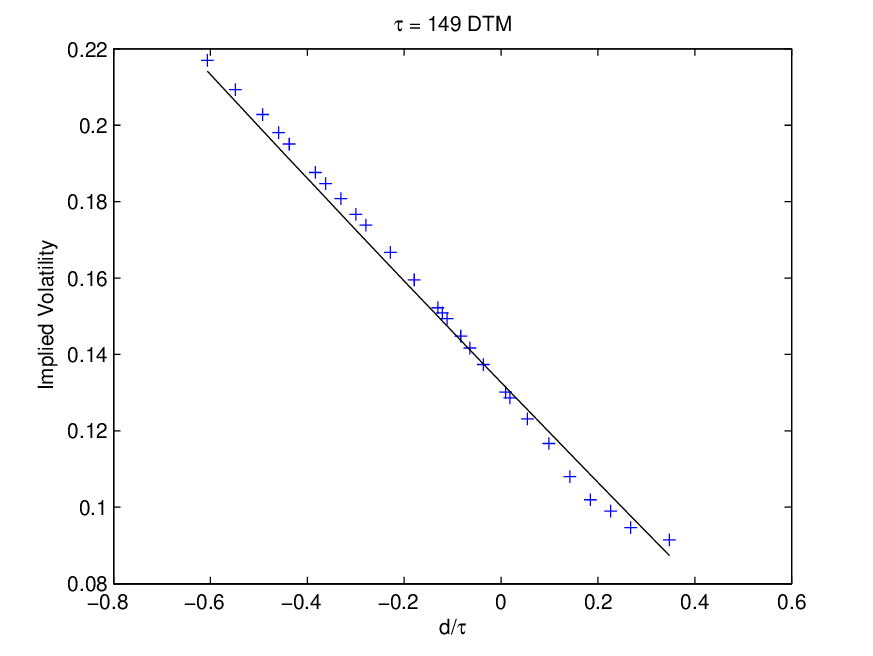} &
\includegraphics[width=0.5\textwidth,height=0.23\textheight]{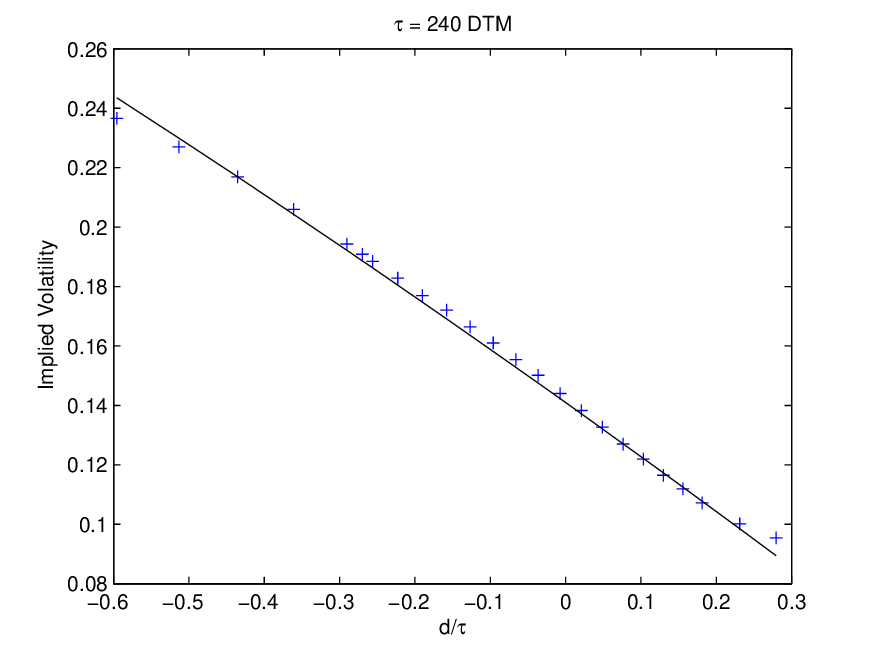} \\ \hline
\includegraphics[width=0.5\textwidth,height=0.23\textheight]{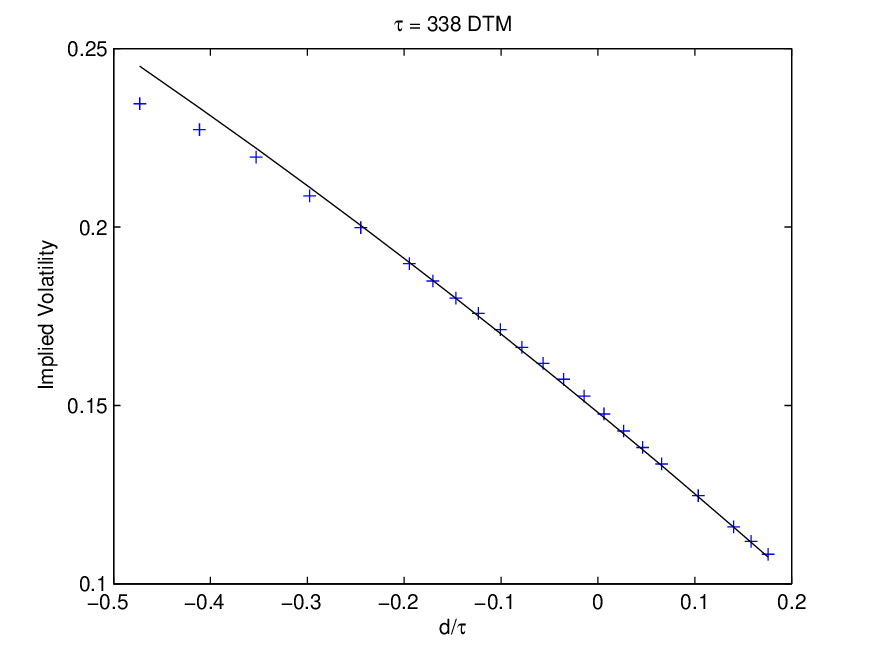} &
\includegraphics[width=0.5\textwidth,height=0.23\textheight]{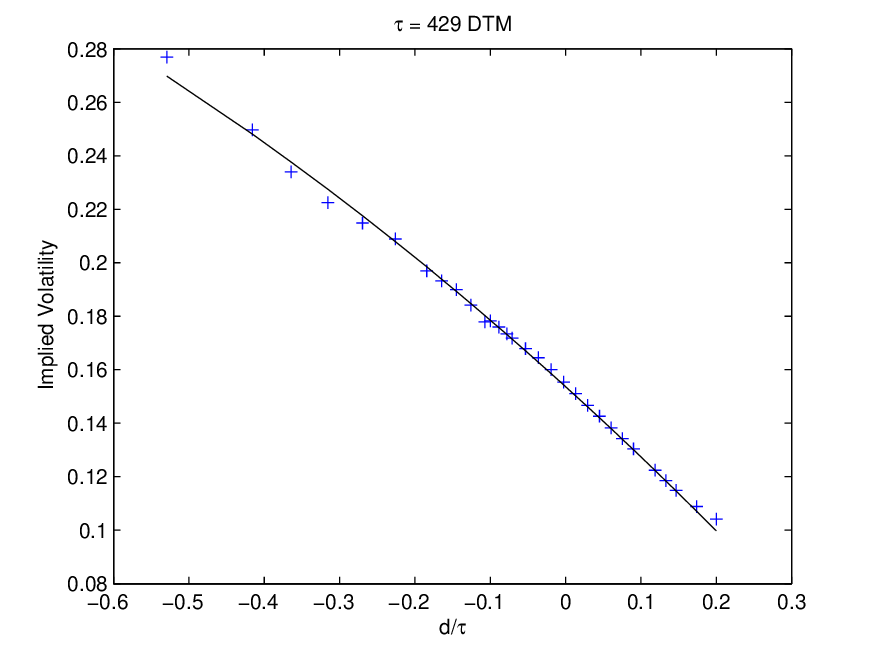} \\ \hline
\includegraphics[width=0.5\textwidth,height=0.23\textheight]{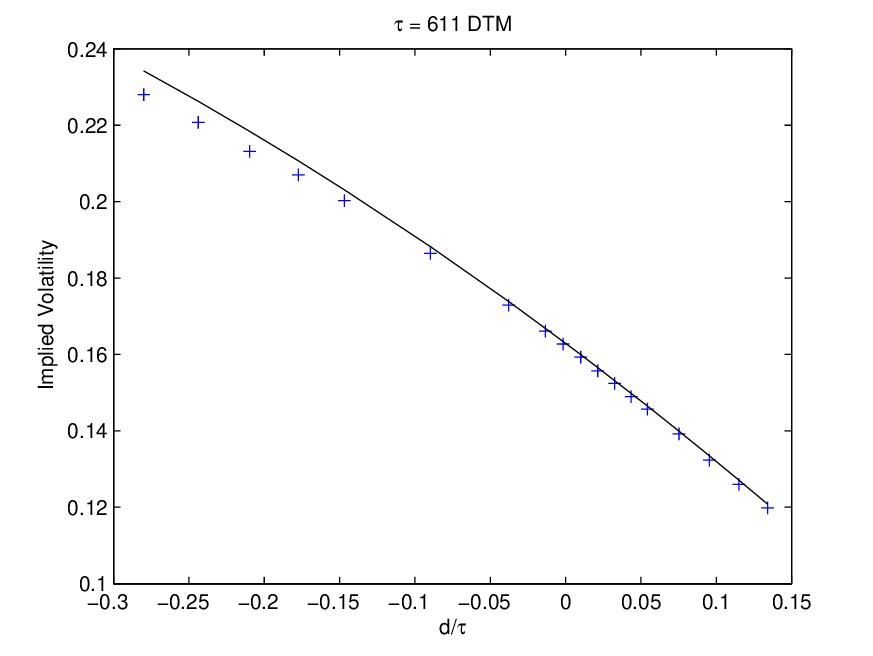} &
\includegraphics[width=0.5\textwidth,height=0.23\textheight]{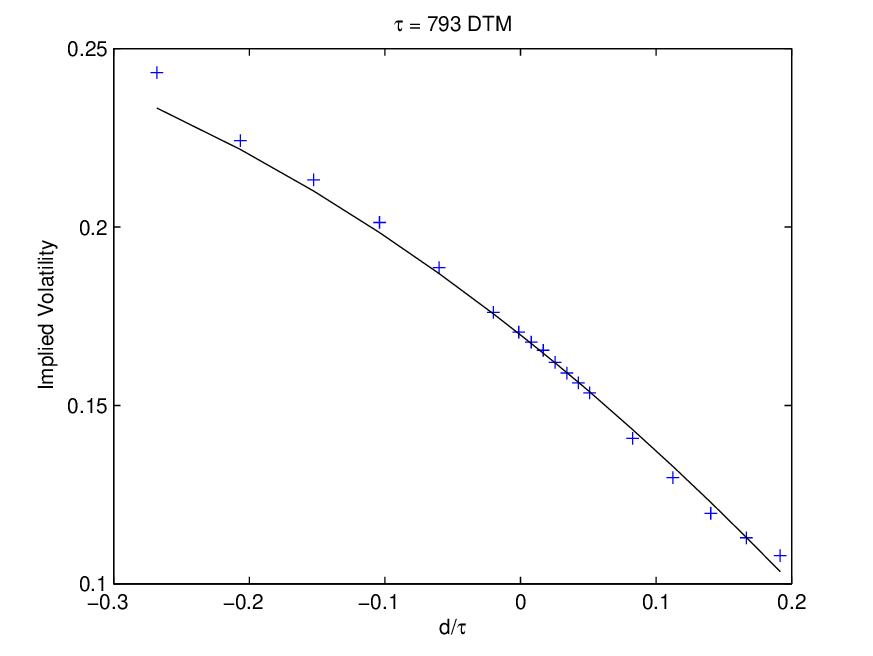} \\ \hline
\end{tabular}
\caption{\small{\emph{Implied volatility fit to S\&P 500 index options on October 19, 2006. Note that this is the result of a single calibration to all maturities and not a maturity-by-maturity calibration. Each panel shows the DTM=days to maturity.}}}
\label{fig:ImpVolFit2006}
\end{figure}

%\clearpage
\begin{figure}
\centering
\begin{tabular}{ | c | c |}
\hline
\includegraphics[width=0.5\textwidth,height=0.23\textheight]{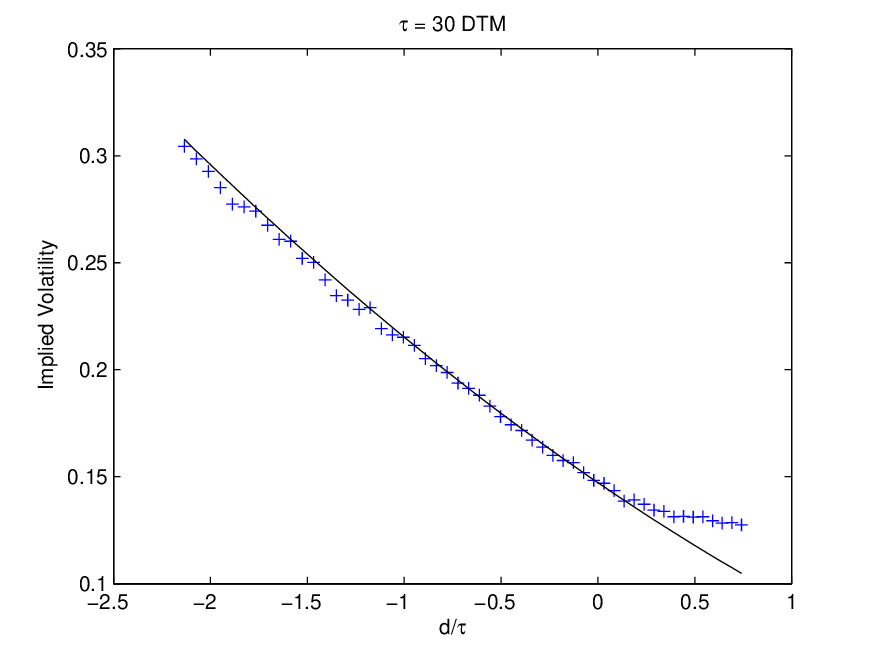} &
\includegraphics[width=0.5\textwidth,height=0.23\textheight]{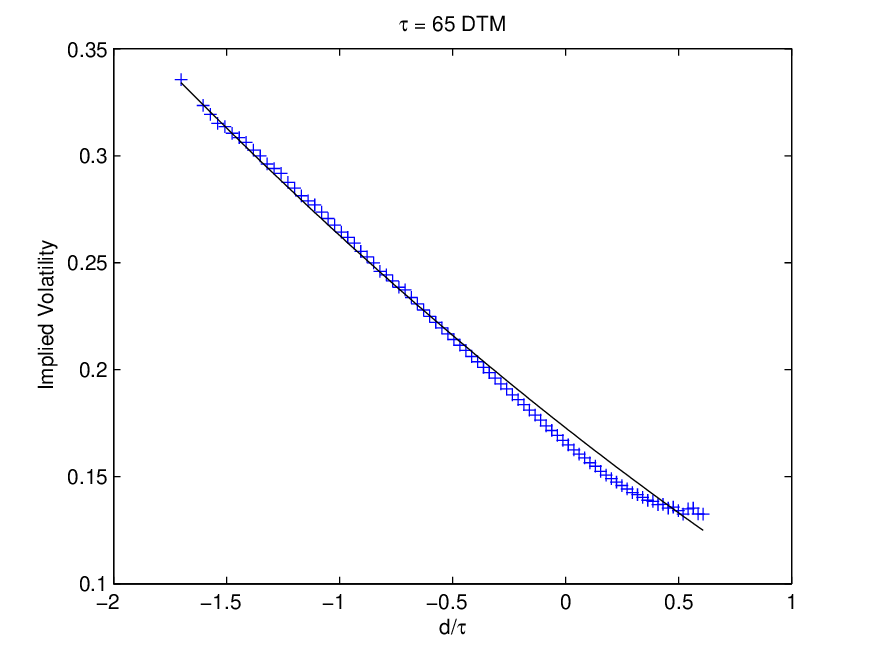} \\ \hline
\includegraphics[width=0.5\textwidth,height=0.23\textheight]{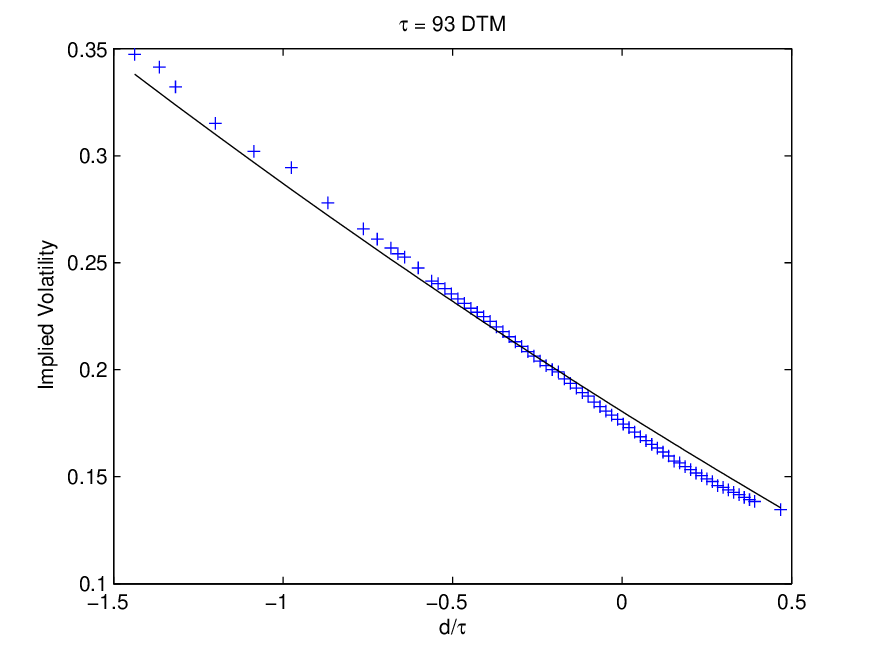} &
\includegraphics[width=0.5\textwidth,height=0.23\textheight]{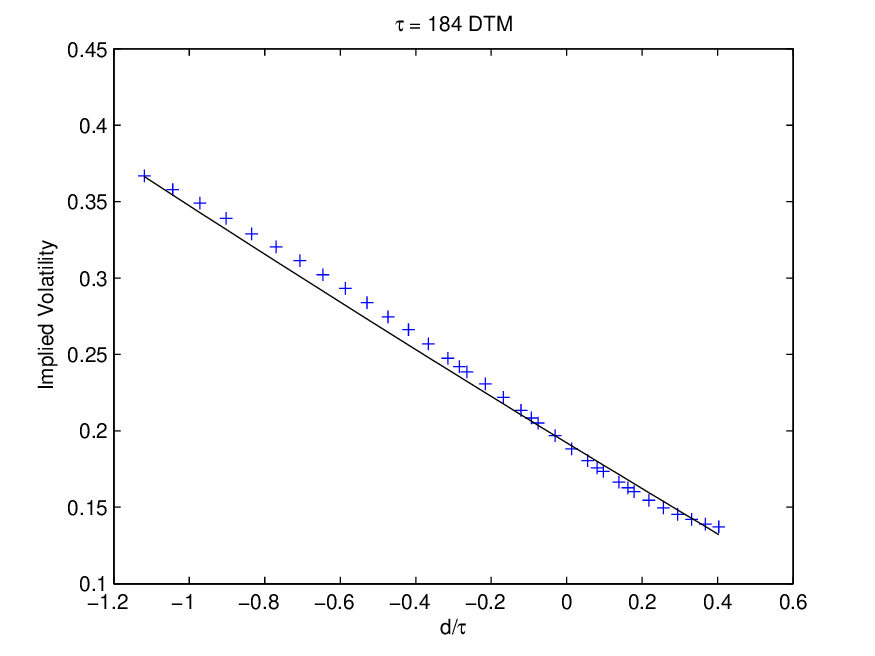} \\ \hline
\includegraphics[width=0.5\textwidth,height=0.23\textheight]{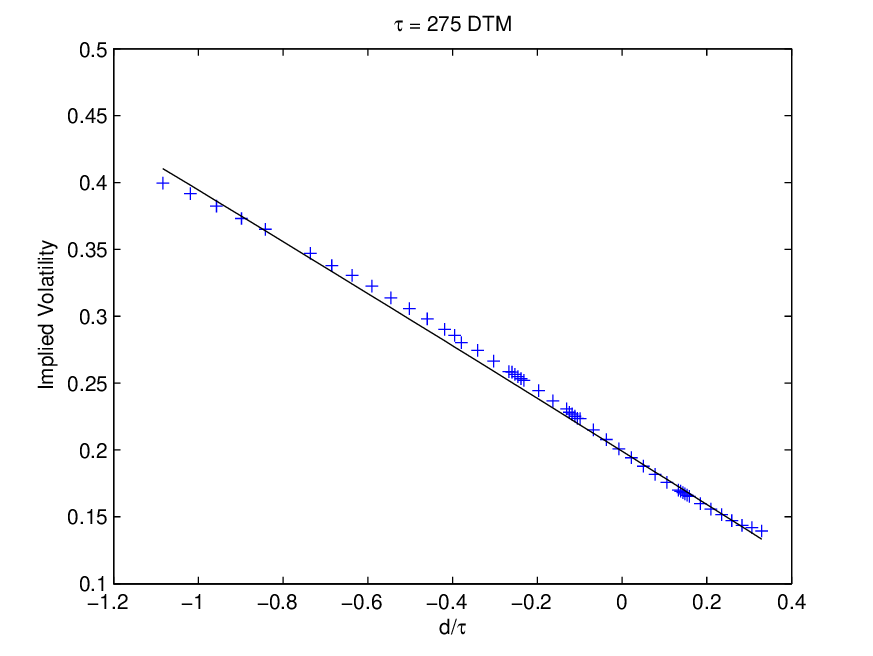} &
\includegraphics[width=0.5\textwidth,height=0.23\textheight]{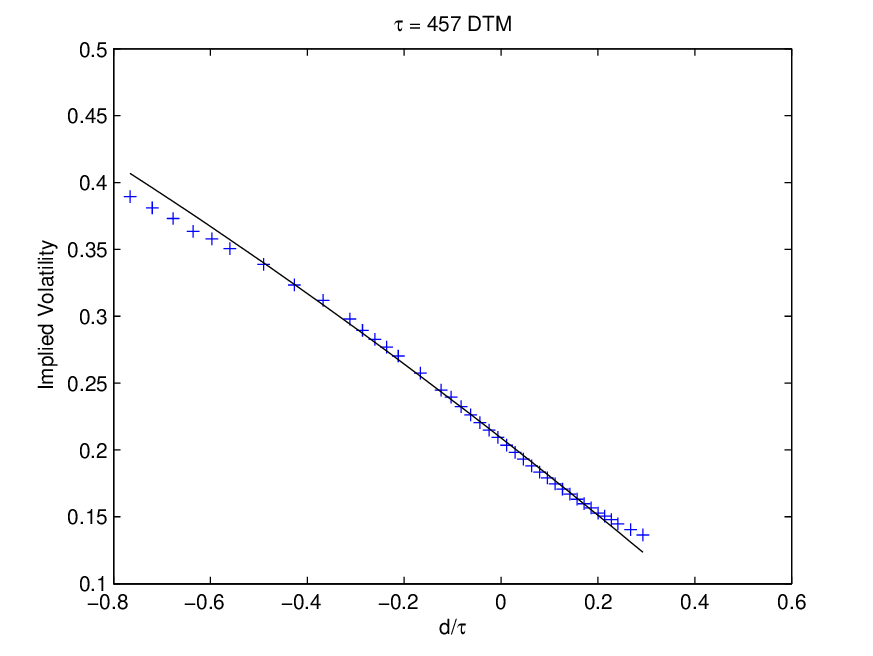} \\ \hline
\includegraphics[width=0.5\textwidth,height=0.23\textheight]{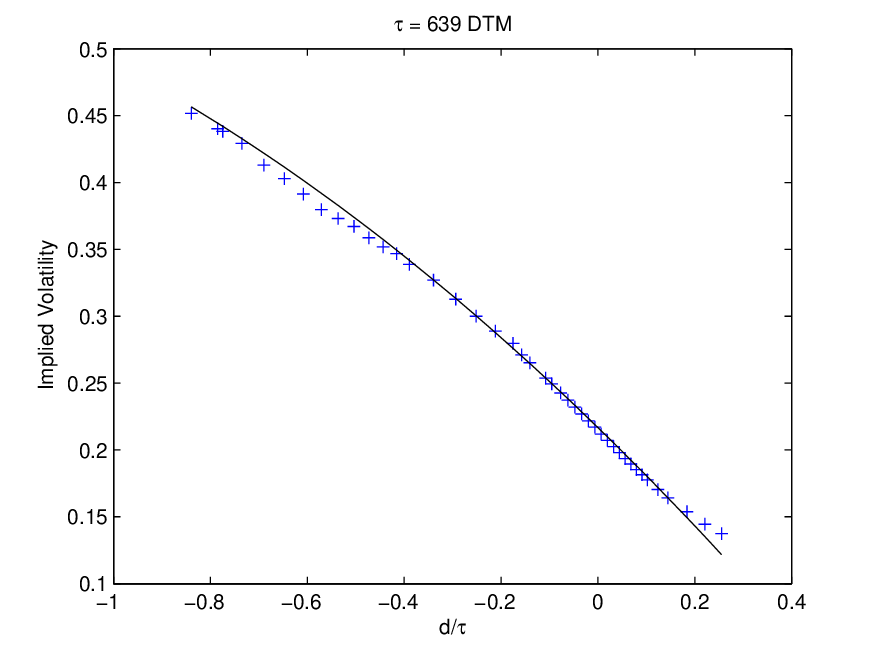} &
\includegraphics[width=0.5\textwidth,height=0.23\textheight]{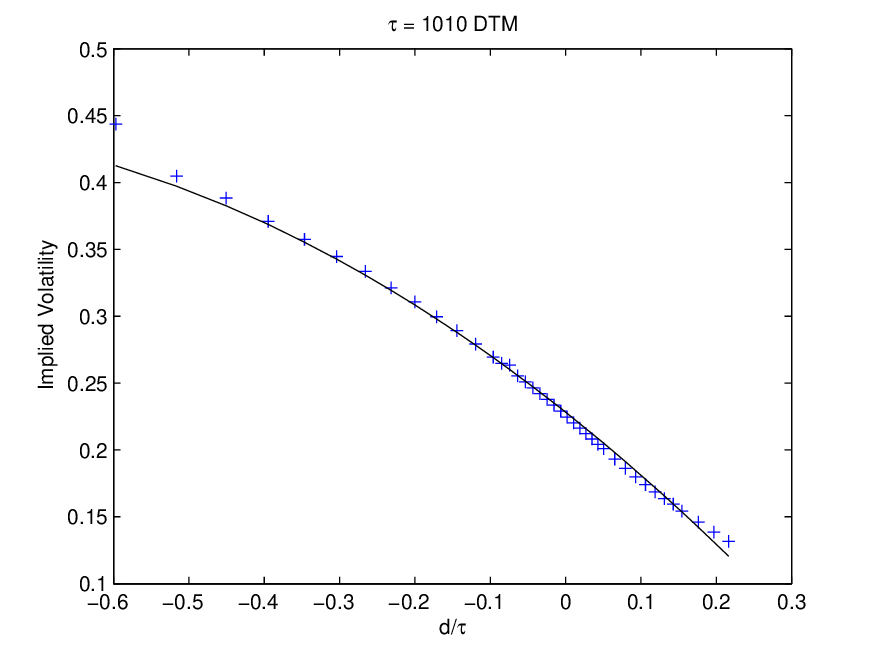} \\ \hline
\end{tabular}
\caption{\small{\emph{Implied volatility fit to S\&P 500 index options on March 18, 2010. Note that this is the result of a single calibration to all maturities and not a maturity-by-maturity calibration.}}}
\label{fig:ImpVolFit2010}
\end{figure}

%\clearpage

%%%%%%%%%%%%%%%%%%%%%%%%%%%%%%%%%%%%%%%
%
%				SECTION: Concluding Remarks
%
%%%%%%%%%%%%%%%%%%%%%%%%%%%%%%%%%%%%%%%

\section{Concluding Remarks}\label{sec:conclusion}
We have derived a second order asymptotic approximation for European options under multiscale stochastic volatility models with fast and slow factors. Proof of convergence requires a terminal layer analysis that is developed probabilistically, in contrast to the techniques of matched asymptotic expansions that are more common in fluid mechanics. The price approximation is translated to an implied volatility surface approximation which is quadratic in log-moneyness and highly nontrivial in the term structure direction. 
We have shown that the second order approximation fits the data well across strikes and maturities (Figures \ref{fig:ImpVolFit2006} and \ref{fig:ImpVolFit2010}). Moreover, the extracted parameters are small when they should be small in the regime of the asymptotic analysis (Section \ref{datasec}).

\clearpage

%%%%%%%%%%%%%%%%%%%%%%%%%%%%%%%%%%%%%%%
%
%				SECTION: Appendix
%
%%%%%%%%%%%%%%%%%%%%%%%%%%%%%%%%%%%%%%%

\appendix

%%%%%%%%%%%%%%%%%% 	Accuracy %%%%%%%%%%%%%%%%%%%%

\section{Proof of Accuracy for Smooth Payoffs}\label{appendix:proof}
In this appendix, we derive the accuracy result for options with smooth payoffs $h$ 
as described in Remark \ref{item:cases} following 
Assumption \ref{itm:h}.
This is needed in order to give a meaning to the terminal value $P_{2,0}(T,x,y,z)$ studied in Section \ref{sec:P2} and justify the regularization argument for general payoffs given in Appendix \ref{sec:call}.

In what follows, we will make use of the following Lemma several times. 
\begin{lemma}\label{D1}
Let $h$ be a smooth payoff function,
that is $h$ is $C^\infty(0,\infty)$, and it and all its derivatives grow at most polynomially at $0$ and $\infty$. 
Then its Black-Scholes price $P_{BS}(\tau,x;\sigma)$ is also $C^\infty(0,\infty)$ in $x$, and  $\d_x^kP_{BS}$ ($k\geq 0$)  are also at most polynomially growing at $0$ and $+\infty$ in the current stock price $x$, and is bounded uniformly in $\tau\in[0,T]$ for fixed $x>0$.

\label{niceGreeks}
\end{lemma}
\begin{proof}
From the formula
\eqref{PBSdef}, 
we see that $P_{BS}$ is $C^\infty(0,\infty)$ in $x$, and grows at most polynomially in $x$ at $0$ and $+\infty$ as inherited from the behavior of $h$. Then, we compute
\begin{align}
\d_x^kP_{BS}(\tau,x;\sigma) = e^{-r\tau}\int
h^{(k)}\(xe^{(r-\frac12\sigma^2)\tau + \sigma\sqrt{\tau}\,\xi}\)
\(e^{(r-\frac12\sigma^2)\tau + \sigma\sqrt{\tau}\,\xi}\)^k\frac{e^{-\xi^2/2}}{\sqrt{2\pi}}\,d\xi,
\end{align}
where $h^{(k)}$ is the $k$-the derivative of $h$, which is at most polynomially growing by assumption, and therefore $\d_x^kP_{BS}$ is also at most polynomially growing at $0$ and $+\infty$ in $x$, and uniformly bounded in $\tau\in[0,T]$ for fixed $x>0$.
\end{proof}
We note that this Lemma does not hold for the nonsmooth case of puts and calls where the derivatives of the payoff are singular at the strike price.

\begin{remark}
Since we have $P_{0,0}(t,x,z)=P_{BS}(T-t,x;\sigb(z))$, it follows that $\D_kP_{0,0}=x^k\partial^k_{x}P_{0,0}$ is at most polynomially growing in $x$ and bounded uniformly in $\tau\in[0,T]$ for fixed $x>0$.
\label{Dkgrowth}
\end{remark}
We will also use the fact that $Y$ and $Z$ have moments of all orders uniformly bounded in $\eps$ and $\del$ (thanks to Assumptions \ref{item:Ybound} and \ref{item:Zbound} made on $Y^{(1)}$ and $Z^{(1)}$ in Section \ref{sec:assumptions}):
\begin{lemma}\label{LmD1}
If $J(y,z)$ is at most polynomially growing, then for every $(y,z)$  there exists a positive constant $C<\infty$ such that
\begin{align}\label{eq:momentsrecall}
\sup_{t\leq T}\sup_{\eps,\del \leq 1} \EEE \[  |J(Y_t,Z_t)| \mid Y_0 = y, Z_0 = z \] \leq C .
\end{align}
\end{lemma}
The proof of this lemma can be found following Lemma 4.9 in \cite{fpss}.\\[-0.5em]

The following property will also be used in what follows:
\begin{lemma}\label{LmD2}
For each $k\in \ZZ$, there exists a constant $C_k<\infty$ depending on $x$ and $T$ such that
\begin{align}\label{eq:momentsX}
\sup_{t\leq T}\sup_{\eps,\del \leq 1} \EEE \[  |X_t|^k \mid X_0=x, Y_0 = y, Z_0 = z \] \leq C_k .
\end{align}
\end{lemma}
\begin{proof}
This is a simple consequence of \eqref{eq:RiskNeutral} and the boundedness of $f(y,z)$ (Assumption \ref{item:fbounded} of Section \ref{sec:assumptions}):
\begin{align*}
 |X_t|^k&=% X_t^k=
 x^k \exp\(krt-\frac{k}{2}\int_0^tf^2(Y_s,Z_s)ds+k\int_0^tf(Y_s,Z_s) dW_s^{\star(0)}\) \\
&=x^k \exp\(krt+\frac{k^2-k}{2}\int_0^tf^2(Y_s,Z_s)ds -\frac{k^2}{2}\int_0^tf^2(Y_s,Z_s)ds+k\int_0^tf(Y_s,Z_s) dW_s^{\star(0)}\) \\
&\leq x^k \exp\(krt+\frac{k^2-k}{2}\overline{c}^2t-\frac{k^2}{2}\int_0^tf^2(Y_s,Z_s)ds+k\int_0^tf(Y_s,Z_s) dW_s^{\star(0)}\),
\end{align*}
where $\overline{c}$ is the upper bound on the volatility function $f$ in Assumption \ref{item:fbounded}.
Therefore, $$\EEE \[ |X_t|^k \] \leq x^k \exp\(krt+\frac{k^2-k}{2}\overline{c}^2t\). $$
\end{proof}

\subsection{Intermediate Lemmas\label{AppA1}}
\begin{lemma}\label{expcon} Let $\xi(x,z)$ and $\chi(y,z)$ be functions that are at most polynomially growing, with $\< \chi(\cdot,z) \> = 0$ for all $z$. Assume further that $\xi(x,z)$ is smooth in $(x,z)$ with derivatives at most polynomially growing and $\chi(y,z)$ is smooth in $z$ with derivatives at most polynomially growing as well. Then we have that
\begin{align}
\EEE_{t,x,y,z} \[ \chi(Y_T,Z_T) \xi(X_T,Z_T)\]=\O(\eps^{q/2}+ \sqrt{\delta})\quad \mbox{for $q<1$}.\label{ld4}
\end{align}
\end{lemma}
In order to establish Lemma \ref{expcon}, we will need the following.

\begin{lemma}\label{nightmare} Let $\chi(y,z)$ be a function that is at most polynomially growing, with $\< \chi(\cdot,z) \> = 0$ for all $z$. Then, for $q<1$ and $z$ fixed, there exists $\bar{\eps}>0$ and a polynomial $C(y)$ such that
$$
\left |\EEE_{t,y} [\chi(Y_s,z) | Y_{s-\eps^q} ] \right |\leq \sqrt{\eps} \left |C(Y_{s-\eps^q})\right |\,\quad\mbox{for any} \quad 0<\eps\leq\bar{\eps}\, \quad\mbox{and} \quad s\geq t+\eps^q.
$$
\end{lemma}
The proof of Lemma \ref{nightmare} is given at the end of this section. 

\begin{proof}[Proof of Lemma \ref{expcon}]
First, we replace $Z_T$ with $z=Z_t$.  This replacement results in an $\O(\sqrt{\del})$ error:
\begin{align}
\EEE_{t,x,y,z} \[ \chi(Y_T,Z_T) \xi(X_T,Z_T)\] - \EEE_{t,x,y,z} \[ \chi(Y_T,z) \xi(X_T,z)\]
&=	\O(\sqrt{\del}) . \label{eq:z.error}
\end{align}
To see this, we observe from \eqref{eq:RiskNeutral} that
\begin{align}
Z_T
&=	z + \del \int_t^T c(Z_s) ds - \sqrt{\del} \int_t^T \Gam(Y_s,Z_s) g(Z_s) ds + \sqrt{\del} \int_t^T g(Z_s) dW_s^{\star(2)} .
\end{align}
The error \eqref{eq:z.error} is then deduced by Taylor expanding $\chi(y,z) \xi(x,z)$ with respect to $z$ and using the linear growth of coefficients in Assumption \ref{item:strong} in Section \ref{sec:assumptions}, the polynomial growth of functions $\chi$, $\xi$ and their derivatives, and the uniform finiteness of moments of all orders in Lemma \ref{LmD1}.

Next, we replace $X_T$ by $X_{T-\eps^q}$ where $q<1$.  This results in an $\O(\eps^{q/2})$ error:
\begin{align}
\EEE_{t,x,y} \[ \chi(Y_T,z) \xi(X_T,z)\] - \EEE_{t,x,y,z} \[ \chi(Y_T,z) \xi(X_{T-\eps^q},z)\]
&=	{\O(\eps^{q/2}) }. \label{eq:x.error}
\end{align}
The error \eqref{eq:x.error} is deduced by using \eqref{eq:RiskNeutral} to write
\begin{align}
X_T
&=	X_{T-\eps^q} + r \int_{T-\eps^q}^T X_s ds + \int_{T-\eps^q}^T f(Y_s,Z_s) X_s {dW_s^{\star(0)} } ,
\end{align}
and then by Taylor expanding $\xi(x,z)$ about the point $x=X_{T-\eps^q}$, and once again using that $\xi(x,z)$ and its derivatives are at most polynomially growing in $x$ and the moments estimate in 
Lemma \ref{LmD2}.

Now, observe that
\begin{align}
\EEE_{t,x,y,z} \[ \chi(Y_T,z) \xi(X_{T-\eps^q},z) \]
&=	\EEE_{t,x,y,z} \[ \xi(X_{T-\eps^q},z) \EEE [\chi(Y_T,z) |\F_{T-\eps^q} ] \] \\
&=
\EEE_{t,x,y,z} \[ \xi(X_{T-\eps^q},z) \EEE [\chi(Y_T,z) |Y_{T-\eps^q} ] \] . \label{eq:conditional} 
\end{align}
Using Lemma \ref{nightmare} at $s=T$, polynomial growth and moment estimates, we deduce that the expectation in \eqref{ld4} is $\O(\eps^{q/2}+ \sqrt{\delta})$ for $q<1$.
\end{proof}

\subsection{Proof of Theorem \ref{thm:accuracy}\label{smoothproof}\label{AppA2}}
Now, we recall our price approximation $\Pt^{\eps,\del}$ from \eqref{eq:Phat}:
\begin{align}
P^{\eps,\del}\approx \Pt^{\eps,\del}
&=	P_{0,0} + \sqrt{\eps} \, P_{1,0} + \sqrt{\del} \, P_{0,1} + \sqrt{\eps \,\del} \, P_{1,1} + \eps \, P_{2,0}
+ \del \, P_{0,2} ,\label{Ptilde}
\end{align}
where $\{P_{i,j}, i+j\leq 2\}$ are given in Proposition \ref{thm:prices}.
The singular perturbation proof involves terms with higher order in $\eps$, and so we introduce
\begin{align}
\Ph^{\eps,\del}
&=	\Pt^{\eps,\del}+\eps^{3/2} P_{3,0} + \eps^2 P_{4,0}+\eps\sqrt{\del}\,P_{2,1}+\eps^{3/2}\sqrt{\del}\,P_{3,1}. \label{phat}
\end{align}	
\begin{remark}\label{poisson7}
The additional terms $P_{3,0},\,P_{4,0},\,P_{2,1},\,P_{3,1}$ are solutions of the Poisson equations \eqref{eq:u3poisson}, \eqref{eq:u4poisson}, \eqref{eq:u21poisson} and \eqref{eq:u31poisson} whose centering conditions have been used to obtain lower order terms in the price expansion. Since these four additional terms are not part of our approximation, but instead are used only for the proof of accuracy, we simply need them to be any solution of these four Poisson equations, which are all of the form 
$$ \L_0 P = \sum_{k\geq1} c_k(t,y,z)\D_kP_{0,0}, $$
where the sum is \emph{finite}, the $c_k(t,y,z)$ are at most polynomially growing in $y$ and $z$, and bounded uniformly in $t\in[0,T]$, and the $\D_kP_{0,0}$ are at most polynomially growing in $x$, and  bounded uniformly in $t\in[0,T]$ for fixed $x>0$ by Remark \ref{Dkgrowth}.
Therefore, by Assumption \ref{item:Poisson}, the solutions $P_{3,0},\,P_{4,0},\,P_{2,1},\,P_{3,1}$ are at most polynomially growing in $(x,y,z)$, and are bounded uniformly in $t\in[0,T]$.
\end{remark}
Next, we define the residual
\begin{align}	
R^{\eps,\del}
&:=	P^{\eps,\del} - \Ph^{\eps,\del} .
\end{align}
The proof of Theorem \ref{thm:accuracy} consists of showing that $R^{\eps,\del}=\O(\eps^{1+q/2}+\eps\sqrt{\del}+\del\sqrt{\eps}+\delta^{3/2})$ for $q<1$. 
By the choices made in Sections \ref{sec:one}, \ref{sec:one.half} and \ref{Odel}, 
when applying the operator $\L^{\eps,\del}$ to the function $ R^{\eps,\del}$,
all of the terms of order $\eps^{-1},\eps^{-1/2},1,\eps^{1/2},\eps, \del^{1/2}\eps^{-1},\del^{1/2}\eps^{-1/2},\del^{1/2},\del^{1/2}\eps^{1/2},\del\eps^{-1},\del\eps^{-1/2},\del$
cancel, as does the term $\L^{\eps,\del}P^{\eps,\del}$.  Hence,
we deduce that the residual $R^{\eps,\del}$ satisfies the following PDE:
\begin{align}\label{eq:Rpde}
\L^{\eps,\del} R^{\eps,\del} + \S^{\eps,\del}
&=0 ,
\end{align}
pointwise in $(t,x,y,z)$, where  
the source term $\S^{\eps,\del} $in \eqref{eq:Rpde} is quite lengthy to write explicitly. However, it is straightforward to check that it is a {finite} sum of the form
\begin{align}\label{dummylabel}
\S^{\eps,\del}=\sum_{i,j:\, i + j \geq 3}
\sqrt{\eps}^{\,i} \sqrt{\del}^{\,j} 
\sum_{k \geq 1} C_{i,j,k}(t,y,z) \D_k P_{0,0} ,
\end{align}
where the coefficients $C_{i,j,k}(t,y,z)$ are bounded uniformly in  $t \in [0,T]$ and at most polynomially growing in $y$ and $z$. We know the terms $\D_k P_{0,0}$, are at most polynomially growing in $x$ and bounded uniformly in $t\in[0,T]$ for fixed $x$ by Lemma \ref{niceGreeks} and the observation in Remark \ref{Dkgrowth}.
Consequently the source term in \eqref{eq:Rpde} is at most polynomially growing in $x, y$ and $z$, uniformly bounded in $t\in[0,T]$ and $\eps,\del\leq 1$. Thus we have $\S^{\eps,\del}=\O(\eps^{3/2}+\eps\sqrt{\del}+\del\sqrt{\eps}+\delta^{3/2})$.

Using the terminal conditions for $\{P_{i,j}, i+j\leq 2\}$, we deduce the terminal condition for the residual:
\begin{align}\label{eq:Rbc}
R^{\eps,\del}(T,x,y,z)
&=	-\eps P_{2,0}(T,x,y,z) + \S^{\eps,\del}_T,
\end{align}
pointwise in $(x,y,z)$, where, again, the terms in $\S^{\eps,\del}_T$ come from the Poisson equations discussed in Remark \ref{poisson7}. It is straightforward to check that $ \S^{\eps,\del}_T$ is of the form 
\begin{align}
\S^{\eps,\del}_T(x,y,z)=\sum_{i,j:\, i + j \geq 3}
\sqrt{\eps}^{\,i} \sqrt{\del}^{\,j} 
\sum_{k \geq 1} C_{i,j,k}(y,z) \D_k h(x), \label{SedT}
\end{align}
where again the sum is finite and the coefficients $C_{i,j,k}(y,z)$ are at most polynomially growing in $y$ and $z$.
The terms $\D_kh(x)$, are at most polynomially growing in $x$ by the assumption in Theorem
\ref{thm:accuracy}. Consequently the term $\S^{\eps,\del}_T$ in \eqref{eq:Rbc} is 
at most polynomially growing in $x, y$ and $z$, uniformly in $\eps,\del\leq 1$. Thus we have $\S^{\eps,\del}_T=\O(\eps^{3/2}+\eps\sqrt{\del}+\del\sqrt{\eps}+\delta^{3/2})$. 
The same polynomial growth condition holds for 
\begin{align}\label{P2explicit}
P_{2,0}(T,x,y,z)=-\frac{1}{2}\phi(y,z)\D_2P_{0,0}(T,x,z)=-\frac{1}{2}\phi(y,z)\D_2h(x).
\end{align}
It is important to note that the non-vanishing terminal value $P_{2,0}(T,x,y,z)$ plays a particular role since it appears at the $\eps$ order.
The probabilistic representation of $R^{\eps,\del}$, solution to the Cauchy problem \eqref{eq:Rpde}-\eqref{eq:Rbc}, is therefore
\begin{align}
R^{\eps,\del}(t,x,y,z)
&=	\frac{\eps}{2}\, \EEE_{t,x,y,z} \[ e^{-r(T-t)} \phi(Y_T,Z_T) \D_2h(X_T)\] + \O(\eps^{3/2}+\eps\sqrt{\del}+\del\sqrt{\eps}+\delta^{3/2}), \label{eq:R}
\end{align}
where $\EEE_{t,x,y,z}$ denotes expectation under the $(\eps,\del)$-dependent dynamics \eqref{eq:RiskNeutral} starting at time $t<T$ from $(x,y,z)$. 
The term denoted by $\O(\eps^{3/2}+\eps\sqrt{\del}+\del\sqrt{\eps}+\delta^{3/2})$ comes from $\S^{\eps,\del}$ in \eqref{eq:Rpde} and $\S^{\eps,\del}_T$ in \eqref{eq:Rbc} and it retains the same order because of the uniform control of the moments of $X$, $Y$ and $Z$ recalled in Lemmas \ref{LmD1} and \ref{LmD2} at the beginning of this section.
We next examine the above expectation in \eqref{eq:R} detail.

From Lemma \ref{expcon} with $\xi= \D_2h$ and $\chi=\phi$, 
where 
smoothness in $z$ follows from the smoothness of $f$ (Assumption \ref{item:smooth} in Section \ref{sec:assumptions}), 
we have
$$\EEE_{t,x,y,z} \[ e^{-r(T-t)} \phi(Y_T,Z_T) \D_2h(X_T)\]=\O(\eps^{q/2}+ \sqrt{\delta})\quad \mbox{for $q<1$},$$
by our choice \eqref{centeringphi}.
We then conclude from \eqref{eq:R} that $R^{\eps,\del}$ is $\O(\eps^{1+q/2}+\eps\sqrt{\del}+\del\sqrt{\eps}+\delta^{3/2})$ for any $q < 1$, which establishes Theorem \ref{thm:accuracy}. 
\begin{remark}\label{Rmk5}
This is exactly where we see that our choice of terminal condition \eqref{eq:u2BC} for $P_{2,0}$, which leads to  \eqref{centeringphi}, was necessary because if $\< \phi(\cdot,z) \> \ne 0$, then 
the expectation in \eqref{eq:R} would be of order $1$ and 
the residual would be of order $\eps$. 
\end{remark}

\subsection{Proof of Lemma \ref{nightmare}}\label{AppA3}
Let us first consider the case $\Lambda=0$. For $z$ fixed, $\chi(y,z)$ being at most polynomially growing in $y$, there exists $a>0$ and an integer $k$ such that $|\chi(y,z)|\leq a(y^{2k}+1)$. By Assumption \ref{item:gap} in Section \ref{sec:assumptions}, we are in position to apply Theorem 6.1 of \cite{MeynTweedie} (note that by assuming Feller property and a strict positive density for the invariant distribution,  the condition that every petite set is compact for some skeleton chain is satisfied). Therefore,
there exists $b<\infty$ and $\lambda>0$ such that
$$
|\EEE_{y}[\chi(Y^{(1)}_t,z)]-\langle \chi(\cdot,z)\rangle|=|\EEE_{y}[\chi(Y^{(1)}_t,z)]|\leq ab(y^{2k}+1)e^{-\lambda t} \quad \mbox {for every } \, t.
$$
By stationarity one deduces that for $s-\eps^q\geq 0$,
$$
|\EEE_{y,s-\eps^q}[\chi(Y_s,z)]-\langle \chi(\cdot,z)\rangle|=|\EEE_y[\chi(Y^{(1)}_{1/\eps^{1-q}},z)]|\leq ab(y^{2k}+1)e^{-\lambda /\eps^{1-q}},
$$
and consequently
$$
|\EEE_{t,y}[\chi(Y_s,z)|Y_{s-\eps^q}]|\leq ab(Y_{s-\eps^q}^{2k}+1)e^{-\lambda /\eps^{1-q}}.
$$
Lemma \ref{nightmare} follows by using $e^{-\lambda /\eps^{1-q}}\leq \sqrt{\eps}$ for $\eps\leq 1$. Note that this last inequality is what we need for the second order  accuracy  studied in this paper but can be improved (in fact, to any power of $ \eps$ up to a multiplicative constant or for $\eps$ small enough).

However, under the pricing measure $\PPP$, due to the presence of the possibly nonzero market price of volatility risk $\Lam(y)$,
we need to deal with the perturbed infinitesimal generator $\L_0-\sqrt{\eps}\beta(y)\Lambda(y)\partial_y$ and its associated diffusion process denoted by $Y^{(1,\eps)}_t$ which satisfies
\begin{align}
dY^{(1,\eps)}_t
&=	\( \alpha(Y^{(1,\eps)}_t) - \sqrt{\eps} \beta(Y^{(1,\eps)}_t) \Lam(Y_t^{(1,\eps)}) \) dt + \beta(Y^{(1,\eps)}_t)\,dW_t^{\star(1)} , &
Y_0^{(1,\eps)}
&=	y . \label{eq:Y.again}
\end{align}
The process $Y^{(1,\eps)}$ in \eqref{eq:Y.again} admits the invariant distribution $\Pi_\eps$ with density
\begin{align}
\pi_\eps(y)
&=	\frac{J_\eps}{\beta^2(y)}\exp\(2 \int_0^y \frac{\alpha(u) - \sqrt{\eps}\beta(u)\Lam(u)}{\beta^2(u)}du \) ,
\end{align}
where $J_\eps$ is a normalization factor. Using Assumption \ref{item:gapeps} and following the argument given above in the case $\Lambda=0$, we obtain that there exists $b<\infty$ and $\lambda>0$ independent of $\eps\leq 1$ such that
$$
|\EEE_{t,y}[\chi(Y_s,z)|Y_{s-\eps^q}]-\<\chi(\cdot,z)\>_\eps|\leq ab(Y_{s-\eps^q}^{2k}+1)e^{-\lambda /\eps^{1-q}}\leq ab(Y_{s-\eps^q}^{2k}+1)\sqrt{\eps}.
$$
 Now, expanding $\pi_\eps$ (including $J_\eps$), we derive for any $g \in L_1(\Pi_\eps)$
\begin{align}
\<g\>_\eps
&=  \<g\> - 2 \sqrt{\eps} \< \(\int_0^\cdot \frac{\Lam(u)}{\beta(u)}du\) \( g(\cdot) - \<g\> \)\> + {\O(\eps)}.
\label{eq:<phi>eps}
\end{align}
Hence, using the fact that $\<\chi(\cdot,z)\>=0$ and the triangle inequality, Lemma \ref{nightmare} follows. Note that the term in $\sqrt{\eps}$ in \ref{eq:<phi>eps} would generate a contribution of order $\sqrt{\eps}$ from $P_2$ which would contribute a term of order $\eps^{3/2}$ if one would seek an expansion of the price at that order.

%%%%%%%%%%%%%%%%%%% Call Options  %%%%%%%%%%%%%%%%%%

\section{Proof of Theorem \ref{thm:accuracy} \label{sec:call}}
In this appendix, we consider payoffs $h$ satisfying Assumption \ref{itm:h}.
We regularize such a payoff $h$ by replacing it with its Black-Scholes price with time to maturity $\gam>0$ and volatility $\sigb(z)$ which appears as a constant volatility, $z$ being a parameter.
Accordingly, we define {
\begin{align}
h^\gam(x,z)
	&=	P_{BS}(\gam,x;\sigb(z)),\label{eq:hD}
\end{align}
where $P_{BS}(\tau,x;\sigma)$ is  the Black-Scholes price of an option with payoff $h$ as a function of the time to maturity $\tau$, the stock price $x$, and the volatility $\sigma$.  }
We note that, for $\gam>0$, the regularized payoff $h^\gam$, as a function of $x$, is $C^\infty$, at most polynomially growing at $0$ and $+\infty$ as well as its derivatives.  As such, $h^\gam$ is smooth, as considered in Appendix \ref{appendix:proof}.

The price $P^{\eps,\delta,\gam}(t,x,y,z)$ of the option with the regularized payoff satisfies
\begin{align}
\L^{\eps,\delta} P^{\eps,\delta,\gam}
	&=	0 , &
P^{\eps,\delta,\gam}(T,x,y,z)
	&=	h^\gam(x,z) ,
\end{align}
where the operator $\L^{\eps,\delta}$ is given in \eqref{eq:L,eps,del}.
Corresponding to the price approximation $\Pt^{\eps,\del}$ given in \eqref{eq:Phat}, we introduce the second order approximation of the regularized option price denoted by $\Pt^{\eps,\del,\gam}$:
\begin{align}
\Pt^{\eps,\del,\gam}
	&=	P_{0,0}^\gam + \sqrt{\eps} P_{1,0}^\gam +  \sqrt{\delta} P_{0.1}^\gam+\eps P_{2,0}^\gam+ \sqrt{\eps}\sqrt{\delta} P_{1,1}^\gam +\delta P_{0,2}^\gam,\label{PtildeDelta}
\end{align}
where, from Proposition \ref{thm:prices}, $P_{0,0}^\gam$ is the Black-Scholes price  of the option maturing at $T$ with payoff $h^\gam(x,z)$, evaluated at volatility $\sigb(z)$. 
{Since we have regularized the payoff in \eqref{eq:hD} by using the Black-Scholes price with volatility $\sigb(z)$, it follows that  $P_{0,0}^\gam$ is given by  
\begin{align}
P_{0,0}^\gam(t,x,z)=P_{0,0}(t-\gam,x,z)
	&=	P_{BS}(T-t+\gam,x;\sigb(z)). \label{P0D}
	\end{align}
Similarly, the other terms in \eqref{PtildeDelta} are solutions of the PDE problems in \eqref{eq:key} with $h$ replaced by $h^\Delta$, and they are given explicitly in Proposition \ref{thm:prices}.}
Note that the term $\eps P_{2,0}^\gam$ in \eqref{PtildeDelta} plays a particular role. From \eqref{P2explicit}, it is given by
\begin{align}
\eps P_{2,0}^\gam(t,x,y,z)=-\frac{1}{2}\eps \phi(y,z)\D_2P^\gam_{0,0}(t,x,z),\label{P2Delta}
\end{align}
where $\phi$ is centered, and at maturity, this term becomes $-\frac{1}{2}\eps\phi(y,z)\D_2h^\gam(x,z)$.

The proof of Theorem \ref{thm:accuracy} will rely on the following three Lemmas, which we prove below.
\begin{lemma}
\label{thm:P}
For a fixed point $(t,x,y,z)$ with $t<T$, there exist constants $\gamb_1>0$, $\epsb_1 > 0$ and $c_1>0$ such that
\begin{align}
| P^{\eps,\delta}(t,x,y,z) - P^{\eps,\delta,\gam}(t,x,y,z) | \leq c_1 \gam ,
\end{align}
for all $0< \gam \leq \gamb_1$ and $0 < \eps \leq \epsb_1$.
\end{lemma}
Lemma \ref{thm:P} controls the error between the model price  and the model price with the regularized payoff.
\begin{lemma}
\label{thm:Q}
For a fixed point $(t,x,y,z)$ with $t<T$, there exist constants $\gamb_2>0$, $\epsb_2 > 0$ and $c_2>0$ such that
\begin{align}
| \Pt^{\eps,\del}(t,x,y,z) - \Pt^{\eps,\del,\gam}(t,x,y,z) | \leq c_2 \gam ,
\end{align}
for all $0< \gam \leq \gamb_2$ and $0 < \eps \leq \epsb_2$.
\end{lemma}
Lemma \ref{thm:Q} controls the error between the approximated price and the approximated price with the regularized payoff.
\begin{lemma}
\label{thm:PQ}
For a fixed point $(t,x,y,z)$ with $t<T$, there exist constants $\gamb_3>0$, $\epsb_3 > 0$ and $c_3>0$ such that
\begin{align}
| P^{\eps,\delta,\gam}(t,x,y,z) - \Pt^{\eps,\del,\gam}(t,x,y,z) | \leq c_3 
\( \eps^{1+q/2}  +\eps \sqrt{\delta}+\delta\sqrt{\eps}+\delta^{3/2}\),
\end{align}
for all $0 < \eps \leq \epsb_3$, any $q<1$, and uniformly in $\Delta\leq\gamb_3$.
\end{lemma}
Lemma \ref{thm:PQ} controls the error between the model price and the approximated price, both  with the regularized payoff.

\subsection{Proof of Theorem \ref{thm:accuracy}}
The proof follows directly from Lemmas \ref{thm:P}, \ref{thm:Q} and \ref{thm:PQ}.  Take $\epsb = \min(\epsb_1,\epsb_2,\epsb_3)$ and choose $\Delta=\eps^{3/2}$.
Then, using Lemmas \ref{thm:P}, \ref{thm:Q} and \ref{thm:PQ}, we find
\begin{align}
|P^{\eps,\delta} - \Pt^{\eps,\del}|
	&\leq	|P^{\eps,\delta} - P^{\eps,\delta,\gam}| + | P^{\eps,\delta,\gam} - \Pt^{\eps,\del,\gam}| + | \Pt^{\eps,\del,\gam}- \Pt^{\eps,\del} | \\
	&\leq 2 \max(c_1,c_2) \eps^{3/2} + c_3 \( \eps^{1+q/2}  +\eps \sqrt{\delta}+\delta\sqrt{\eps}+\delta^{3/2}\) ,\\
	&=\O(\eps^{3/2-}+\eps\sqrt{\del}+\del\sqrt{\eps}+\delta^{3/2}),
\end{align}
where the functions are evaluated at a fixed $(t,x,y,z)$ with $t<T$.

\subsection{Proofs of Lemmas \ref{thm:P} and \ref{thm:Q}}\label{sec:lemmas}
\begin{proof}[Proof of Lemma \ref{thm:P}]
The proof is a straightforward extension of {\cite[Lemma 4.1]{fouque2003proof}}. It requires a multi-factor ``correlated Hull-White formula"  with general payoffs which is in \cite[Section 2.5.4]{fpss}.
We give some details here since it introduces notations that will also be used in the proof of Lemma \ref{thm:4} below. 
Conditioning on the volatility path $(Y_u,Z_u)_{t\leq u\leq T}$ (or their driving brownian motions $(W_u^{\star(1)},W_u^{\star(2)})_{t\leq u\leq T}$), we obtain the representations
\begin{align}
&P^{\eps,\delta}(t,x,y,z)= \EEE_{t,x,y,z} \{P_{BS}(t,xe^{\zeta_{t,T}}; \sigb_{\perp,t,T})\},\\
&P^{\eps,\delta,\gam}(t,x,y,z)=\EEE_{t,x,y,z} \{P_{BS}(t,xe^{\zeta_{t,T}+r\Delta}; \sigb^\Delta_{\perp,t,T})\},
\end{align}
where $P_{BS}$ is the Black-Scholes price with payoff $h$ and maturity $T$, and for $t< s\leq T$:
\begin{align}
&\zeta_{t,s}=\rho_1\int_t^sf(Y_u,Z_u)dW_u^{\star(1)}+\rho_2\int_t^sf(Y_u,Z_u)dW_u^{\star(2)}-\half(\rho_1^2+\rho_2^2)  \int_t^sf(Y_u,Z_u)^2du , \label{def:zeta} \\
&\sigb_{\perp,t,s}^2=\frac{c_0^2}{s-t} \int_t^sf(Y_u,Z_u)^2du, \quad \mbox{with}\quad 0< c_0^2:=
\frac{1-\rho_1^2-\rho_2^2-\rho_{12}^2+2\rho_1\rho_2\rho_{12}}{1-\rho_{12}^2}\leq 1,\label{def:sigperp}\\
&(\sigb^\Delta_{\perp,t,s})^2=\sigb_{\perp,t,s}^2+\frac{\Delta\sigb^2(z)}{s-t}.\nonumber
\end{align}
Therefore,
\begin{align}
 | P^{\eps,\delta}(t,x,y,z) - P^{\eps,\delta,\gam}(t,x,y,z) | &\leq e^{-r(T-t)} \EEE_{t,x,y,z} \(\int |h(xe^{\eta'+\zeta_{t,T}})|\cdot |p_1(\eta')-p_2(\eta')|d\eta'\),
\end{align}
where $p_1$ denote the Gaussian density of  ${\cal N}\((r-\frac{\sigb_{\perp,t,T}^2}{2})(T-t), \sigb_{\perp,t,T}^2(T-t)\)$ and $p_2$ the Gaussian density of ${\cal N}\(r\Delta+(r-\frac{(\sigb^\Delta_{\perp,t,T})^2}{2})(T-t), (\sigb^\Delta_{\perp,t,T})^2(T-t)\)$. Observe that the variance $\sigb_{\perp,t,T}^2(T-t)$ is bounded and bounded below by 
$c_0^2\underline{c}^2(T-t)$  where $0<\underline{c}\leq f(y,z)$ from Assumption \ref{item:fbounded} in Section \ref{sec:assumptions}. Lemma \ref{thm:P} follows by using polynomial growth of $h$ at $0$ and $\infty$, and exponential moments of $\zeta_{t,T}$.
\end{proof}

\begin{proof}[Proof of Lemma \ref{thm:Q}]
{From Proposition \ref{thm:prices}, we can express each $P_{i,j}$ ($i+j\leq 2$) as an operator acting on $P_{BS}(T-t,x;\sigb(z))$, and since derivatives with respect to $\sigma$ can be converted to derivatives with respect to $x$ by the 
Vega-Gamma relation \eqref{eq:Dsig}, we can write the price approximation $\Pt^{\eps,\del}$ in \eqref{eq:Phat} as $\Pt^{\eps,\del}(t,x,y,z)={\cal G}P_{BS}(T-t,x;\sigb(z))$, where the operator ${\cal G}$ is a polynomial in $\{\D_i\}$ with bounded coefficients for $(y,z)$ given. Similarly we can express $\Pt^{\eps,\del,\gam}$ as $\Pt^{\eps,\del,\gam}={\cal G}P_{BS}(T-t+\Delta,x;\sigb(z))$, and therefore 
\begin{align}
\Pt^{\eps,\del}- \Pt^{\eps,\del,\gam}
	&=	
	{\cal G}\(P_{BS}(T-t,x;\sigb(z))-P_{BS}(T-t+\Delta,x;\sigb(z))\).
\end{align}
Using 
the differentiability of $P_{BS}$ and $\{\D_i\} P_{BS}$ with respect to $t$ at $t<T$,  Lemma \ref{thm:Q} follows easily.}
\end{proof}

\subsection{Estimates on Greeks}
The key to proving Lemma \ref{thm:PQ} is the following Lemma providing uniform estimates.
\begin{lemma}\label{thm:4}
{As in Lemma \ref{thm:PQ}, in what follows, $t$ is fixed such that $t<T$.} 
{Let $\chi(y,z)$ be a function which is at most polynomially growing in $(y,z)$, and  that is smooth in $z$ with partial derivatives with respect to $z$ that are at most polynomially growing in $(y,z)$.}
Denote by $\eta_s=\log X_s$ the log-process and by $\eta=\log x$ the corresponding log-variable. Then, 
for any integer $k$,
there exists a finite constant $c > 0$,
which may depend on $(t,x,y,z,T)$, such that  
uniformly in $\eps$, $\delta$, $\Delta >0$ and $t\leq s\leq T$:
\begin{align}
 \Big|  {\EEE_{t,x,y,z}\left[ \chi(Y_s,Z_s)\d_\eta^k P_{0,0}^\gam(s,e^{\eta_s},Z_s)\right] } \Big| 
	&\leq c ,  \label{eq:first.bound0} \\
\end{align}
and, for a given $p\geq 0$, 
\begin{align}
 \Big|  \EEE_{t,x,y,z} \int_t^T (T-s)^p  e^{-r(s-t)} \chi(Y_s,Z_s) \d_\eta^k P_{0,0}^\gam(s,e^{\eta_s},Z_s) ds \Big| 
	&\leq c .
		\label{eq:second.bound0}
\end{align}
Additionally, if $\chi$ is centered, $\< \chi(\cdot,z) \> =0$ for all $z$, then, for any $q<1$ and any  integer $k$,  there exists a finite constant $c > 0$, which may depend on $(t,x,y,z,T)$ such that for any $\eps$ {satisfying} $\eps^q\leq T-t$ and any $s$ {satisfying} $t+\eps^q\leq s \leq T$, uniformly in $\Delta>0$
we have 
\begin{align}
 \Big|  {\EEE_{t,x,y,z}\left[ \chi(Y_s,Z_s) \d_\eta^k P_{0,0}^\gam(s,e^{\eta_s},Z_s)\right] } \Big| 
	&\leq c (\eps^{q/2}+\sqrt{\delta}) , \label{eq:first.bound}
\end{align}
and, for a given $p\geq 0$,
\begin{align}
\Big|  \EEE_{t,x,y,z} \int_t^T (T-s)^p  e^{-r(s-t)} \chi(Y_s,Z_s) \d_\eta^k P_{0,0}^\gam(s,e^{\eta_s},Z_s) ds \Big| 
	&\leq c (\eps^{q/2}+ \sqrt{\delta}).
		\label{eq:second.bound}
\end{align}
\end{lemma}
\begin{proof}[Proof of Lemma \ref{thm:4}]
This is an improved version of Lemma 5.2 in \cite{fouque2003proof} where the proof consisted in an explicit computation of $\d_\eta^k P_{0,0}^\gam$ in the case of a call payoff. Here, we aim at estimates which are uniform in $\Delta$. Conditioning on the volatility path $(Y_u,Z_u)_{t\leq u\leq s}$ and using the notations introduced in the proof of Lemma \ref{thm:P} in Section \ref{sec:lemmas}, we get:
\begin{align}
\EEE_{t,x,y,z}\left[ \chi(Y_s,Z_s)\d_\eta^k P_{0,0}^\gam(s,e^{\eta_s},Z_s)\right] &=
\EEE_{t,x,y,z}\left[ \chi(Y_s,Z_s)\int h(e^{\eta'+\zeta_{t,s}})\d_\eta^k p(\eta'-\eta)d\eta'\right] ,\label{eq:density}
\end{align}
where $p$ is 
the Gaussian density of 
\begin{align}
&{\cal N}\((r-\half\sigb_{\perp,t,s}^2)(s-t)+(r-\half \sigb(Z_s)^2)(T+\Delta-s), \sigb_{\perp,t,s}^2(s-t)+\sigb(Z_s)^2(T+\Delta-s)\),\label{density}
\end{align}
and, $\zeta_{t,s}$ and $\sigb_{\perp,t,s}^2$ are defined for $s>t$ in \eqref{def:zeta} and \eqref{def:sigperp} respectively. Note that for $s=t$,  $\zeta_{t,t}=0$ and the Gaussian distribution is simply 
$
{\cal N}\((r-\half \sigb(z)^2)(T+\Delta-t), \sigb(z)^2(T+\Delta-t)\).
$
The uniform bound \eqref{eq:first.bound0} follows from the uniform lower bound of the variance of $p$, polynomial growth of $h$, uniform moments of $Y$ and $Z$ (Lemma \ref{LmD1}), and exponential moments of $\zeta_{t,s}$.
The bound \eqref{eq:second.bound0} is a direct consequence of \eqref{eq:first.bound0}.

If, in addition $\chi$ is centered, we 
define
$$\xi_s=\EEE\left[ \d_\eta^k P_{0,0}^\gam(s,e^{\eta_s},Z_s)\mid (Y_u,Z_u)_{t\leq u\leq s}\right],$$ 
and we  write for $s\geq t+\eps^q$,
\begin{align}
&\EEE_{t,x,y,z}\left[ \chi(Y_s,Z_s) \d_\eta^k P_{0,0}^\gam(s,e^{\eta_s},Z_s)\right] =\EEE_{t,x,y,z}\left[ \chi(Y_s,Z_s) \xi_s\right],\nonumber\\
&\hskip 1cm =\EEE_{t,x,y,z}\left[ \chi(Y_s,Z_s)( \xi_s-\xi_{s-\eps^q})\right]+\EEE_{t,x,y,z}\left[ \xi_{s-\eps^q}\EEE[\chi(Y_s,Z_s)\mid {\cal F}_{s-\eps^q}]\right]\label{twoterms}
\end{align} 
The second term $\EEE_{t,x,y,z}\left[ \xi_{s-\eps^q}\EEE[\chi(Y_s,Z_s)\mid {\cal F}_{s-\eps^q}]\right]$ in \eqref{twoterms}  is treated as in the proof of Lemma \ref{expcon}. Replacing $\EEE[\chi(Y_s,Z_s)\mid {\cal F}_{s-\eps^q}]$ with $\EEE[\chi(Y_s,z)\mid {\cal F}_{s-\eps^q}]$ results in an ${\cal O}(\sqrt{\delta})$ error. Lemma \ref{nightmare} (using the centering condition) and the argument given above to prove \eqref{eq:first.bound0} give
\begin{align}
\big|\EEE_{t,x,y,z}\left[ \xi_{s-\eps^q}\EEE[\chi(Y_s,Z_s)\mid {\cal F}_{s-\eps^q}]\right]\big|&\leq c(\sqrt{\eps}+\sqrt{\delta}).\label{secondterm}
\end{align}
Regarding the first term $\EEE_{t,x,y,z}\left[ \chi(Y_s,Z_s)( \xi_s-\xi_{s-\eps^q})\right]$ in \eqref{twoterms}, we write as in \eqref{eq:density}
\begin{align*}
\EEE_{t,x,y,z}\left[ \chi(Y_s,Z_s)( \xi_s-\xi_{s-\eps^q})\right]&=
\EEE_{t,x,y,z}\left[ \chi(Y_s,Z_s)\int h(e^{\eta'+\zeta_{t,s}})\d_\eta^k \(p-\tilde{p})(\eta'-\eta\)d\eta'\right] ,
\end{align*}
where $p$ is the Gaussian  density of \eqref{density} and $\tilde{p}$ is the Gaussian density of
\begin{align}
&{\cal N}\(-\zeta_{\tilde{s},s}+(r-\half\sigb_{\perp,t,\tilde{s}}^2)(\tilde{s}-t)+(r-\half \sigb(Z_{\tilde{s}})^2)(T+\Delta-\tilde{s}), \sigb_{\perp,t,\tilde{s}}^2(\tilde{s}-t)+\sigb(Z_{\tilde{s}})^2(T+\Delta-\tilde{s})\),\label{densitytilde}
\end{align}
where $\tilde{s}=s-\eps^q$. Using differentiability with respect to the mean and variance of a normal density (with variance bounded away from zero), and, as in the proof of \eqref{eq:first.bound0}, polynomial growth of $h$, uniform moments of $Y$ and $Z$ (Lemma \ref{LmD1}), and exponential moments of $\zeta_{t,s}$, we deduce that
\begin{align}
\big|\EEE_{t,x,y,z}\left[ \chi(Y_s,Z_s)( \xi_s-\xi_{s-\eps^q})\right]\big|&\leq c\,\eps^{q/2}.\label{firstterm}
\end{align}
Combining \eqref{secondterm} and \eqref{firstterm} with $q<1$ gives \eqref{eq:first.bound}.

The uniform bound \eqref{eq:second.bound} follows easily by decomposing the integral over $[t,T]$ into two integrals, one over $[t,t+\eps^q]$ and using the bound \eqref{eq:first.bound0}, and the other one over $[t+\eps^q, T]$ and using the bound \eqref{eq:first.bound}. Note that the factor $(T-s)^p$ in the integral is simply uniformly bounded by $(T-t)^p$.
\end{proof}

\subsection{Proof of Lemma \ref{thm:PQ}\label{AppB4}}
The proof essentially follows the proof of Theorem \ref{thm:accuracy} in Appendix \ref{smoothproof}.
We define the residual $R^{\eps,\delta,\gam}$ for the regularized payoff via the following equation
\begin{align}
P^{\eps,\delta,\gam}
&=\Pt^{\eps,\del,\gam}+\eps^{3/2} P_{3,0}^\gam + \eps^2 P^\gam_{4,0}+\eps\sqrt{\delta}P^\gam_{2,1}+\eps^{3/2}\sqrt{\delta}P^\gam_{3,1}+ R^{\eps,\delta,\gam}, \label{def:Repsdelgam}
\end{align}
where the approximation $\Pt^{\eps,\del,\gam}$ is given by \eqref{PtildeDelta}, and, as in the proof in the smooth case in Section \ref{smoothproof}, we have introduced the additional terms $(P_{3,0}^\gam, P^\gam_{4,0}, P^\gam_{2,1}, P^\gam_{3,1})$. As we discussed in Remark \ref{poisson7} in that section, they are solutions of the Poisson equations \eqref{eq:u3poisson}, \eqref{eq:u4poisson}, \eqref{eq:u21poisson} and \eqref{eq:u31poisson} (augmented with the $\gam$ superscript), whose centering conditions have been used to obtain lower order terms in the price expansion.

More precisely, applying the operator $\L^{\eps,\delta}$ to $R^{\eps,\delta,\gam}$,
we find the analog of \eqref{eq:Rpde}:
\begin{align}
\L^\eps R^{\eps,\delta,\gam}
	&=	{G^{\eps,\gam} + J^{\eps,\del,\gam} },
\end{align}
where the source terms $G^{\eps,\gam}$ and $J^{\eps,\del,\gam}$ are given by
\begin{align}
G^{\eps,\gam}
	&=	- \( \eps^{3/2}(\L_1 P^\gam_{4,0} + \L_2 P^\gam_{3,0}) + \eps^2 \L_2 P_{4,0}^\gam \) , \label{eq:G.eps.gam} \\
J^{\eps,\del,\gam}
	&=	- \sqrt{\del} \Big(
			\eps (\L_2 P^\gam_{2,1}+\L_1 P^\gam_{3,1}+\M_3 P^\gam_{3,0}+ \M_1 P_{2,0}^\gam) %\\ & \qquad \qquad
			+ \eps^{3/2} ( \L_2 P^\gam_{3,1}+\M_1 P^\gam_{3,0}+\M_3 P^\gam_{4,0})
			+ \eps^1 (\M_1 P^\gam_{4,0})
			\Big) \\ &\qquad
			- \del \Big( 
					\sqrt{\eps} ( \M_2 P_{1,0}^\gam+ \M_1 P_{1,1}^\gam+ \M_3 P_{2,1}^\gam )
					+ \eps ( \M_1 P^\gam_{2,1}+ \M_3 P^\gam_{3,1} +  \M_2 P_{2,0}^\gam ) \\ & \qquad \qquad
					+ \eps^{3/2}( \M_2 P^\gam_{3,0} + \M_1 P^\gam_{3,1} )
					+ \eps^2 \M_2 P^\gam_{4,0} 
					\Big)  \\ &\qquad
			- \del^{3/2} \( \M_2 P_{0,1}^\gam+ \M_1 P_{0,2}^\gam + \sqrt{\eps} \M_2  P_{1,1}^\gam +\eps \M_2 P_{2,1}^\gam + \eps^{3/2} \M_2 P_{3,1}^\gam\) 
			\\ & \qquad
			- \del^2  \M_2 P_{0,2}^\gam .
			\label{eq:J.eps.gam.gam}
\end{align}
We have separated the terms involving singular perturbation only, that is $G^{\eps,\gam}$, and the terms involving regular perturbation as well, that is $J^{\eps,\del,\gam}$.
With the same decomposition in mind, at the maturity date $T$, we have
\begin{align}
R^{\eps,\delta,\gam}(T,x,y,z)
	&=	H^{\eps,\gam}(x,y,z) + { K^{\eps,\del,\gam}(x,y,z) },
\end{align}
where the functions $H^{\eps,\gam}$ and $K^{\eps,\del,\gam}$ are given by
\begin{align}
H^{\eps,\gam}(x,y,z)	
	&= - \eps P_{2,0}^\gam(T,x,y,z) - \eps^{3/2} P_{3,0}^\gam(T,x,y,z) - \eps^2 P_{4,0}^\gam(T,x,y,z)  . \label{eq:H.eps.gam} \\
K^{\eps,\del,\gam}(x,y,z)
	&=	{  -\eps\sqrt{\delta}P^\gam_{2,1}(T,x,y,z) - \eps^{3/2}\sqrt{\delta}P^\gam_{3,1}(T,x,y,z) ,}	\label{eq:K.eps.del.gam}
\end{align}
and the particular term $\eps P_{2,0}^\gam(T,x,y,z)$ is given in \eqref{P2Delta}.
The residual $R^{\eps,\delta,\gam}$ has the following stochastic representation
\begin{align}
R^{\eps,\delta,\gam}(t,x,y,z)
	&=	\EEE_{t,x,y,z} \[ - \int_t^T e^{-r(s-t)} G^{\eps,\gam}(X_s,Y_s,Z_s) ds + e^{-r(T-t)} H^{\eps,\gam}(X_T,Y_T,Z_T) \] \\ &
			+{  \EEE_{t,x,y,z} \[ - \int_t^T e^{-r(s-t)} J^{\eps,\del,\gam}(X_s,Y_s,Z_s) ds + e^{-r(T-t)} K^{\eps,\del,\gam}(X_T,Y_T,Z_T) \] } ,\label{rep:Repsdelgam}
\end{align}
At this point, in order to apply the bounds in Lemma \ref{thm:4}, it is useful to change variables to $\eta(x) = \log x$. We note that, for a function $\xi$ that is at least $(n+2m)$-times differentiable, we have
\begin{align}
\D_1^n \D_2^m \xi (\eta(x))
	&=	\sum_{k=n+m}^{n+2m} a_k \d_\eta^{k}\xi(\eta(x)) ,
\end{align}
where the $\{ a_k \}$ are integers.  
Denoting $\tau=T-t$, a direct computation shows that $G^{\eps,\gam}$ is of the form
\begin{align}
G^{\eps,\gam}(t,e^\eta,y,z)
	=& \, \eps^{3/2} \( \sum_{k=1}^5 g_k^{(0)}(y,z) \d_\eta^k + \tau \sum_{k=1}^7 g_k^{(1)}(y,z) \d_\eta^k
			+ \tau^2 \sum_{k=1}^9 g_k^{(2)}(y,z) \d_\eta^k \) P_{0,0}^\gam(t, e^\eta,z) \\ 
		& + \eps^2 \( \sum_{k=1}^6 g_k^{(3)}(y,z) \d_\eta^k + \tau \sum_{k=1}^8 g_k^{(4)}(y,z) \d_\eta^k 
			+ \tau^2 \sum_{k=1}^{10} g_k^{(5)}(y,z) \d_\eta^k \) P_{0,0}^\gam(t, e^\eta,z). 
			\label{eq:G.eps.gam.2}
\end{align}
Likewise, 
one finds that $H^{\eps,\gam}$ is of the form
\begin{align}
{H^{\eps,\gam}(e^\eta,y,z) }
	&=	{ \( \eps \sum_{k=1}^2h_k^{(0)}(y,z) \d_\eta^k
			+ \eps^{3/2} \sum_{k=1}^3 h_k^{(1)}(y,z) \d_\eta^k
			+ \eps^2 \sum_{k=1}^4 h_k^{(2)}(y,z) \d_\eta^k \) P_{0,0}^\gam(T,e^\eta,z) } . \label{eq:H.eps.gam.2} 
\end{align}
where $\langle h_1^{(0)} \rangle = \langle h_2^{(0)} \rangle=0$.
Then, by expressions \eqref{eq:G.eps.gam.2} and \eqref{eq:H.eps.gam.2}, and Lemma \ref{thm:4} (bounds \eqref{eq:first.bound} and \eqref{eq:second.bound} for the terms in $\eps$, and bounds \eqref{eq:first.bound0} and \eqref{eq:second.bound0} for the other terms),
there exists a constant $c>0$ such that uniformly in $\Delta>0$:
\begin{align}
\Big| \EEE_{t,x,y,z} \[ H^{\eps,\gam}(X_T,Y_T,Z_T) \] \Big|
	&\leq  c  (\eps^{1+q/2} +\eps\sqrt{\delta}) ,\label{boundH}\\
	\Big| \EEE_{t,x,y,z} \[ \int_t^T e^{-r(s-t)} G^{\eps,\gam}(X_s,Y_s,Z_s) ds \]\Big|
	&\leq  c  (\eps^{1+q/2} +\eps\sqrt{\delta}) . \label{boundG}
\end{align}
Next, analyzing  the terms $J^{\eps,\del,\gam}$ and $K^{\eps,\del,\gam}$ given by \eqref{eq:J.eps.gam.gam} and \eqref{eq:K.eps.del.gam} respectively, 
we find there exists a constant $c>0$ such that uniformly in $\Delta>0$:
\begin{align}
\Big| \EEE_{t,x,y,z} \[ K^{\eps,\del,\gam}(X_T,Y_T,Z_T) \] \Big|
&\leq  { c \, \eps \sqrt{\del} },\label{boundK}\\
\Big| \EEE_{t,x,y,z} \[ \int_t^T e^{-r(s-t)} J^{\eps,\del,\gam}(X_s,Y_s,Z_s) ds \]\Big|
&\leq   c \( \eps  \sqrt{\del} + \del\sqrt{\eps} + \del^{3/2} \)  . \label{boundJ}
\end{align}
Here, we omit the lengthy details which consist in writing decomposition formulas for $J^{\eps,\del,\gam}$ and $K^{\eps,\del,\gam}$ similar to the ones obtained for $G^{\eps,\gam}$ and $H^{\eps,\gam}$ in \eqref{eq:G.eps.gam.2} and \eqref{eq:H.eps.gam.2}.
 $J^{\eps,\del,\gam}$ and $K^{\eps,\del,\gam}$ correspond to performing first a regular perturbation bringing a factor $\sqrt{\del}$ and then performing a first order singular perturbation which does not involve boundary layer terms.

Putting together the definition  \eqref{def:Repsdelgam}, the representation formula \eqref{rep:Repsdelgam}, and the bounds \eqref{boundH}, \eqref{boundG}, \eqref{boundK}, \eqref{boundJ},  we deduce that
 for fixed $(t,x,y,z)$ with $t<T$, and $q<1$, there exists a constant $c$ such that
\begin{align}
| P^{\eps,\delta,\gam} - \Pt^{\eps,\del,\gam} |
	&=	{ |\eps^{3/2} P_{3,0}^\gam + \eps^2 P^\gam_{4,0}+\eps\sqrt{\delta}P^\gam_{2,1}+\eps^{3/2}\sqrt{\delta}P^\gam_{3,1}+ R^{\eps,\delta,\gam}| } \\
	%&=	|\eps^{3/2}P_{3,0}^\gam + \eps^2 P_{4,0}^\gam + R^{\eps,\gam}+{\cal O}(\eps\sqrt{\delta}+\delta\sqrt{\eps}+\delta^{3/2})| \\
	&\leq c \( \eps^{1+q/2} +\eps \sqrt{\delta}+\delta\sqrt{\eps}+\delta^{3/2} \) ,
\end{align}  
%Note that we have used expressions \eqref{eq:P3.gam} and \eqref{eq:P4.gam} for $P_{3,0}^\gam$ and $P_{4,0}^\gam$ respectively.
{which concludes the proof of Lemma \ref{thm:PQ}.}
%\end{proof}

%%%%%%%%%%%%%%%%%%% Parameter Reduction  %%%%%%%%%%%%%%%%%%

%\section{Proof of Proposition \ref{thm:reduction}}
\section{{Proof of Accuracy after Parameter Reduction in Section \ref{sec:parameter reduction}}}
\label{sec:reduction}
Throughout this Section we use the notation $\O(\eps^{3/2-})$ to indicate terms that are of order $\O(\eps^{1+q/2})$ for any $q<1$.  
Recall from \eqref{eq:replacement} that $\sig^{*2}=\sigb^2+2\sqrt{\eps}V_2$ where, we do not show the $z$-dependence for simplicity of notation. 

We show that %these changes 
replacing $\Pt^{\eps,\del}$ in Theorem \ref{thm:accuracy} by $P^{*,\eps,\del}$ defined in \eqref{Pstardef}
does not alter the order of accuracy of the approximation. Note that we are in fact performing a regular perturbation on the volatility. We provide here a PDE based proof assuming smooth payoffs as in Appendix A and we omit the details of the regularization argument which is a simple application of Lemma \ref{thm:Q} and its extension to the regularization of the approximation $P^{*,\eps,\del}$.

First, we note that $\( P_{0,0} - P_{0,0}^* \)= \O(\sqrt{\eps})$ since
\begin{align}
\< \L_2 \> \( P_{0,0} - P_{0,0}^* \)
		&=	\sqrt{\eps}\,V_2 \D_2 P_{0,0}^* , &
P_{0,0}(T,x,z) - P_{0,0}^*(T,x,z)
		&=	0 . 
\end{align}
Next, we define $E_1^{\eps,\del}(t,x,z)$ by
\begin{align}
E_1^{\eps,\del}
&:= \( P_{0,0} + \sqrt{\eps}P_{1,0} + \sqrt{\del} P_{0,1} \) - \( P_{0,0}^* + \sqrt{\eps}P_{1,0}^* + \sqrt{\del} P_{0,1}^* \),
\end{align}
the difference in the first order approximations.
Note that $E_1^{\eps,\del}(T,x,z)=0$ and
\begin{align}
\<\L_2\> E_1^{\eps,\del}
&=	\[ \sqrt{\eps} \( \V^* + V_2 \D_2 \) + \sqrt{\del} \< \M_1 \> \] \( P_{0,0}^* - P_{0,0} \) 
		+ \eps V_2 \D_2 P_{1,0}^* +  \sqrt{\eps \del} V_2 \D_2 P_{0,1}^*.
\end{align}
Thus, we conclude that $E_1^{\eps,\del} = \O(\eps + \sqrt{\eps \del})$.  

Similarly incorporating the order $\eps$ term, we define $E_2^\eps(t,x,y,z)$ by
\begin{align}
E_2^\eps :=	\( P_{0,0} + \sqrt{\eps}P_{1,0} + \eps P_{2,0} \) - \( P_{0,0}^* + \sqrt{\eps}P_{1,0}^* + \eps P_{2,0}^* \) .
\end{align}
{From equation \eqref{eq:<phi>eps} and by using $\D_2\( P_{0,0} - P_{0,0}^* \) = \O(\sqrt{\eps})$ one can show that $E_2^{\eps}(T,x,y,z)=\O(\eps^{3/2-})$}. We then compute
\begin{align}
\<\L_2\> E_2^\eps 
&=	\sqrt{\eps} \V \[ \( P_{0,0}^* + \sqrt{\eps}P_{1,0}^* \) - \( P_{0,0} + \sqrt{\eps}P_{1,0} \) \] 
		+ \eps \A \( P_{0,0}^* - P_{0,0} \) + \eps^{3/2} V_2 \D_2 P_{2,0}^* .
\end{align}

Incorporating the order $\sqrt{\eps\del}$ term, we define $E_3^\eps(t,x,z)$ by
\begin{align}
E_3^\eps :=	\( P_{0,1} + \sqrt{\eps} P_{1,1} \) - \(  P_{0,1}^* + \sqrt{\eps } P_{1,1}^* \) .
\end{align}
Note that $E_3^\eps(T,x,z)=0$ and
\begin{align}
\<\L_2\> E_3^\eps
&=	\< \M_1 \> \[ \( P_{0,0}^* + \sqrt{\eps} P_{1,0}^* \) - \( P_{0,0} + \sqrt{\eps} P_{1,0}^* \) \] 
		+ \sqrt{\eps} \frac{1}{\sigb'} \C \d_z \( P_{0,0}^* - P_{0,0} \)
		+ \sqrt{\eps} \V \( P_{0,1}^* - P_{0,1} \) .
\end{align}
Now define $E_4^\eps(t,x,z)$ by
\begin{align}
E_4^\eps :=	P_{0,2}  - P_{0,2}^* .
\end{align}
Note that $E_4^\eps(T,x,z)=0$ and
\begin{align}
\<\L_2\> E_4^\eps
&=	\<\M_1\> \( P_{0,1}^* -P_{0,1} \) + \M_2 \( P_{0,0}^* -P_{0,0} \) + \sqrt{\eps} V_2 \D_2 P_{0,2}^* .
\end{align}
Finally,
\begin{align}
\<\L_2\> \( E_2^\eps + \sqrt{\del} E_3^\eps + \del E_4^\eps \)
=&	\( \sqrt{\eps} \V + \sqrt{\del} \< \M_1 \> \) E_1^{\eps,\del}
		+ \eps^{3/2} V_2 \D_2 P_{2,0}^* + \sqrt{\eps} \del V_2 \D_2 P_{0,2}^* \\ 
&+ \( \eps \A + \sqrt{\eps\del} \frac{1}{\sigb'} \C \d_z\) \( P_{0,0}^* - P_{0,0} \)
		+ \del \M_2 \( P_{0,0}^* -P_{0,0} \).
\end{align}
Hence, we conclude
\begin{align}
E_2^\eps + \sqrt{\del} E_3^\eps + \del E_4^\eps
&= \O( \eps^{3/2-} + \eps \sqrt{\del} + \sqrt{\eps} \, \del) .
\end{align}

%%%%%%%%%%%%%%%%%%%%%%%%%%%%%%%%%%%%%%%%%%%%%%%%%%%%
%
%			Bibliography
%
%%%%%%%%%%%%%%%%%%%%%%%%%%%%%%%%%%%%%%%%%%%%%%%%%%%%

\bibliographystyle{plain}
\small{\bibliography{secorder}}

\end{document}